\documentclass[12pt]{article}
\usepackage{amsmath,amssymb,amsthm,color}
\usepackage{cite,psfrag,url,setspace,lineno,subfigure}
\usepackage{fullpage}
\usepackage[mathscr]{eucal}
\usepackage[pdftex]{graphicx} 

\newtheorem{thm}{Theorem}[section]

\newtheorem{lemma}[thm]{Lemma}
\newtheorem{prop}[thm]{Proposition}
\newtheorem{cor}[thm]{Corollary}

\newtheorem{bigthm}{Theorem}   

\theoremstyle{remark}
\newtheorem*{rem}{Remark}

\newcommand{\R}{\mathbb{R}}

\newcommand{\Z}{\mathbb{Z}}
\newcommand{\N}{\mathbb{N}}
\newcommand{\G}{\mathcal{N}}
\newcommand{\A}{\mathcal{A}}
\newcommand{\sB}{\mathscr{B}}
\newcommand{\sC}{\mathscr{C}}
\newcommand{\sd}{{\Sigma\Delta}}

\newcommand{\ignore}[1]{} 

\renewcommand{\thefootnote}{\fnsymbol{footnote}}

\title{Sobolev Duals for Random Frames and \\
$\Sigma\Delta$ Quantization of Compressed Sensing Measurements}
\author{C.S. G\"unt\"urk\footnotemark[1], A. Powell\footnotemark[2],
R. Saab\footnotemark[3],  \"O. Y\i lmaz\footnotemark[3] }

\begin{document}
\footnotetext[1]{
Courant Institute of Mathematical Sciences, New York University.}
\footnotetext[2]{
Vanderbilt University.}
\footnotetext[3]{
University of British Columbia.}

\maketitle

\renewcommand{\thefootnote}{\alph{footnote}}

\begin{abstract}
  Quantization of compressed sensing measurements is typically
  justified by the robust recovery results of Cand{\`e}s, Romberg and
  Tao, and of Donoho.  These results guarantee that if a uniform
  quantizer of step size $\delta$ is used to quantize $m$ measurements
  $y = \Phi x$ of a $k$-sparse signal $x \in \R^N$, where $\Phi$
  satisfies the restricted isometry property, then the approximate
  recovery $x^\#$ via $\ell_1$-minimization is within $O(\delta)$ of
  $x$. The simplest and commonly assumed approach is to quantize each
  measurement independently.  In this paper, we show that if instead
  an $r$th order $\Sigma\Delta$ quantization scheme with the same
  output alphabet is used to quantize $y$, then there is an
  alternative recovery method via Sobolev dual frames which guarantees
  a reduction of the approximation error by a factor of
  $(m/k)^{(r-1/2)\alpha}$ for any $0 < \alpha < 1$, if $m \gtrsim_r k
  (\log N)^{1/(1-\alpha)}$. The result holds with high probability on
  the initial draw of the measurement matrix $\Phi$ from the Gaussian
  distribution, and uniformly for all $k$-sparse signals $x$ that
  satisfy a mild size condition on their supports.
\end{abstract}

\section{Introduction}

Compressed sensing is concerned with when and how sparse signals 
can be recovered exactly or approximately from few linear measurements
\cite{Donoho2006,CRT2006,CandesICM}.
Let $\Phi$ be an $m \times N$ matrix providing the measurements where
$m \ll N$, and $\Sigma_k^N$ denote the space 
of $k$-sparse signals in $\R^N$, $k<m$. A standard objective, after a 
suitable change of basis, is that the mapping $x \mapsto y=\Phi x$ 
be injective on $\Sigma^N_k$. Minimal
conditions on $\Phi$ that offer such a guarantee are well-known
(see, e.g. \cite{CDD2009}) and require 
at least that $m \geq 2 k$. On the other hand, under stricter
conditions on $\Phi$, such as the restricted isometry property (RIP), 
one can recover sparse vectors from their measurements
by numerically efficient methods, such as $\ell^1$-minimization.
Moreover, 
the recovery will also be robust when the measurements are corrupted
\cite{CRT}, cf. \cite{D06};
if $\hat y = \Phi x + e$ where $e$ is any vector such that 
$\|e\|_2 \leq \epsilon$,
then the solution $x^\#$ of the optimization problem
\begin{equation}\label{ell1_eps_prog}
\min \|z \|_1 \mbox{ subject to } \|\Phi z - \hat y \|_2 \leq \epsilon 
\end{equation}
will satisfy $\|x - x^\#\|_2 \lesssim \epsilon$. 

The price paid for these stronger recovery guarantees is the somewhat
smaller range of values available for the dimensional parameters $m$,
$k$, and $N$.  While there are some explicit (deterministic)
constructions of measurement matrices with stable recovery guarantees,
best results (widest range of values) have been found via random
families of matrices. For example, if the entries of $\Phi$ are
independently sampled from the Gaussian distribution
$\G(0,\frac{1}{m})$, then with high probability, $\Phi$ will satisfy
the RIP (with a suitable set of parameters) if $m \sim k \log (\frac{N}{k})$.
Significant effort has been put on understanding the phase transition
behavior of the RIP parameters for other random families, e.g.,
Bernoulli matrices and random Fourier samplers.

\subsection*{Quantization for compressed sensing measurements}
The robust recovery result mentioned above is essential to the
practicality of compressed sensing, especially from an
analog-to-digital conversion point of view.  If a discrete alphabet
$\A$, such as $\A = \delta \Z$ for some step size $\delta > 0$, is to
be employed to replace each measurement $y_j$ with a quantized
measurement $q_j:=\hat y_j \in \A$, then the temptation, in light of
this result, would be to minimize $\|e\|_2=\|y - q\|_2$ over $q \in
\A^m$. This immediately reduces to minimizing $|y_j - q_j|$ for each
$j$, i.e., quantizing each measurement separately to the nearest
element of $\A$, which is usually called Pulse Code Modulation (PCM).

Since $\|y - q\|_2 \leq \frac{1}{2} \delta \sqrt{m}$, the robust
recovery result guarantees that
\begin{equation}\label{err_bound_PCM}
\|x - x^\#_{\mathrm{PCM}}\|_2 \lesssim \delta \sqrt{m}.
\end{equation}
Note that \eqref{err_bound_PCM} is somewhat surprising as
the reconstruction error bound does not improve by increasing
the number of (quantized) measurements; on the contrary, it deteriorates.
However, the $\sqrt{m}$ term is
an artifact of our choice of normalization for the measurement matrix
$\Phi$. In the compressed sensing literature, it is conventional to
normalize a (random) measurement matrix $\Phi$ so that it has unit-norm
columns (in expectation). This is the necessary scaling to achieve 
isometry, and for random matrices it 
ensures that $\mathbb{E} \|\Phi x\|^2 = \|x\|^2$ for any
$x$, which then leads to the RIP through
concentration of measure and finally to the robust recovery result stated in
\eqref{ell1_eps_prog}. On the other hand, this normalization imposes
an $m$-dependent dynamic range for the measurements which scales as
$1/\sqrt{m}$, hence it is
not fair to use the same value $\delta$ for the quantizer resolution
as $m$ increases. In this
paper, we investigate the dependence of the recovery error on the
number of quantized measurements where $\delta$ is independent of $m$.
A fair assessment of this dependence can be made only if the
dynamic range of each measurement is kept constant while increasing
the number of measurements. This suggests that the natural
normalization in our setting should ensure that the entries of the
measurement matrix $\Phi$ are independent of $m$. In the specific case
of random matrices, we can achieve this by choosing the entries of
$\Phi$ standard i.i.d.~random variables, e.g. according to $\G(0,1)$. 
With this normalization of
$\Phi$, the robust recovery result of \cite{CRT}, given above, can be
modified as
\begin{equation} \label{ell1_eps_prog_alt}
\|\hat{y}-y\|_2 \le \epsilon\ \implies \|x-x^\#\|_2 \lesssim
\frac{1}{\sqrt{m}}\epsilon,
\end{equation}
which also replaces \eqref{err_bound_PCM} with
\begin{equation}\label{err_bound_PCM_alt}
\|x - x^\#_{\mathrm{PCM}}\|_2 \lesssim \delta.
\end{equation}
As expected, this error bound does not deteriorate with $m$ anymore. 
In this paper, we will adopt this normalization convention
and work with the standard Gaussian distribution $\G(0,1)$ 
when quantization is involved, but also use the more typical
normalization $\G(0,1/m)$ 
for certain concentration estimates that will be derived in 
Section \ref{main_proof}. 
The transition between these two conventions is of course trivial.



The above analysis of quantization error is based on PCM, which
involves separate (independent) quantization of each measurement.
The vast logarithmic reduction of the
ambient dimension $N$ would seem to suggest that this strategy is
essentially optimal since information appears to
be squeezed (compressed) into few uncorrelated measurements. Perhaps
for this reason, the existing literature on quantization of compressed
sensing measurements focused mainly on alternative reconstruction
methods from PCM-quantized measurements and variants thereof, e.g., \cite{BB,ZBC,JHF,dai901quantized,GoyalCSQ,laska2009democracy}. The only exception
we are aware of is \cite{boufounos2007sigma}, which uses $\Sigma\Delta$
modulation to quantize $x$ \emph{before} the random measurements are made.

On the other hand, it is clear that if (once) the support of the signal
is known (recovered), then the $m$ measurements that have been taken
are highly redundant compared to the maximum $k$ degrees of freedom
that the signal has on its support.  At this point, the signal may be
considered {\em oversampled}. 
However, the error bound \eqref{err_bound_PCM_alt} does not offer 
an improvement of reconstruction accuracy, even if additional samples 
become available. (The RIP parameters of $\Phi$ are likely to 
improve as $m$ increases, but this
does not seem to reflect on the implicit constant factor in
 \eqref{err_bound_PCM_alt} satisfactorily.)
This is contrary to the conventional wisdom
in the theory and practice of oversampled quantization in A/D conversion
where reconstruction error decreases as the sampling rate increases, especially
with the use of quantization algorithms specially geared for the
reconstruction procedure. The main goal of this paper is to show how this can 
be done in the compressed sensing setting as well. 

\subsection*{Quantization for oversampled data}
Methods of quantization have long been studied for oversampled data
conversion. Sigma-delta ($\Sigma\Delta$) quantization (modulation),
for instance, is the dominant method of A/D conversion for audio
signals and relies heavily on oversampling, see
\cite{NST-book,DD,G-exp}. In this setting, oversampling is typically
exploited to employ very coarse quantization (e.g., $1$ bit/sample),
however, the working principle of $\Sigma\Delta$ quantization is
applicable to any quantization alphabet. In fact, it is more
natural to consider $\Sigma\Delta$ quantization as a
``noise\footnote{The quantization error is often modeled as white
  noise in signal processing, hence the terminology. However our
  treatment of quantization error in this paper is entirely
  deterministic.} shaping'' method, for it seeks a quantized signal
$(q_j)$ by a recursive procedure to push the quantization error signal
$y-q$ towards an unoccupied portion of the signal spectrum.  In the
case of bandlimited signals, this would correspond to high frequency
bands.

As the canonical example, the standard first-order $\Sigma\Delta$
quantizer computes a bounded solution $(u_j)$ to the difference
equation
\begin{equation}
(\Delta u)_j := u_j - u_{j-1} = y_j - q_j.
\end{equation}
This can be achieved recursively by choosing, for example, 
\begin{equation}
q_j = \arg\min_{p \in \A} |u_{j-1} + y_j - p|.
\end{equation}
Since the reconstruction of oversampled bandlimited signals can be
achieved with a low-pass filter $\varphi$ that can also be arranged to
be well-localized in time, the reconstruction error $\varphi*(y-q) =
\Delta \varphi * u$ becomes small due to the smoothness of $\varphi$.
It turns out that, with this procedure, the reconstruction error is
reduced by a factor of the oversampling ratio $\lambda$, defined to be
the ratio of the actual sampling rate to the bandwidth of $\varphi$.

This principle can be iterated to set up higher-order $\Sigma\Delta$
quantization schemes. It is well-known that a reconstruction accuracy
of order $O(\lambda^{-r})$ can be achieved (in the supremum norm) if a
bounded solution to the equation $\Delta^r u = y - q$ can be found
\cite{DD} (here, $r\in \N$ is the order of the associated $\sd$
scheme). The boundedness of $u$ is important for practical
implementation, but it is also important for the error bound. The
implicit constant in this bound depends on $r$ as well as
$\|u\|_\infty$. Fine analyses of carefully designed schemes have shown
that optimizing the order can even yield exponential accuracy
$O(e^{-c\lambda})$ for fixed sized finite alphabets $\A$ (see
\cite{G-exp}), which is optimal apart from the value of the constant
$c$. For infinite alphabets, there is no theoretical lower bound for
the quantization error as $\lambda$ increases. (However almost all
practical coding schemes use some form of finite alphabet.)

The above formulation of noise-shaping for oversampled data conversion
generalizes naturally to the problem of quantization of arbitrary
frame expansions, e.g., \cite{BPY}. Specifically, we will consider
finite frames in $\R^k$.  Let $E$ be a full-rank $m \times k$ matrix
and $F$ be any left inverse of $E$.  In frame theory, one refers to
the collection of the rows of $E$ as the analysis frame and the
columns of $F$ as the synthesis (dual) frame. For any $x \in \R^k$,
let $y=Ex$ be its frame coefficient vector, $q \in \A^m$ be
its quantization, and let $\hat x:= Fq$ be its reconstruction using the
dual frame.  Typically $\A^m \cap y + \mathrm{Ker}(F) = \emptyset$, so
we have $\hat x \not= x$. The reconstruction error is given by
\begin{equation}\label{SD-err-1}
  x - \hat x = F(y - q),  
\end{equation}
and the goal of noise shaping amounts to arranging $q$ in such a way
that $y-q$ is close to $\mathrm{Ker}(F)$.

If the sequence $(f_j)_1^m$ of dual frame vectors were known to vary
smoothly in $j$ (including smooth termination into null vector), then
$\Sigma\Delta$ quantization could be employed without much alteration,
e.g., \cite{LPY,BPA}. However, this need not be the case for many
examples of frames (together with their canonical duals) that are used
in practice. For this reason, it has recently been proposed in
\cite{BLPY} to use special alternative dual frames, called Sobolev
dual frames, that are naturally adapted to $\Sigma\Delta$
quantization. It is shown in \cite{BLPY} (see also Section
\ref{sec2}) that for any frame $E$, if a standard $r$th order
$\Sigma\Delta$ quantization algorithm with alphabet $\A = \delta \Z$
is used to compute $q := q_{\mathrm{\Sigma\Delta}}$, then with an
$r$th order Sobolev dual frame $F := F_{\mathrm{Sob},r}$
and $\hat x_{\mathrm{\Sigma\Delta}} := F_{\mathrm{Sob},r} 
q_{\mathrm{\Sigma\Delta}}$, the
reconstruction error obeys the bound
\begin{equation}
\|x - \hat x_{\mathrm{\Sigma\Delta}} \|_2 \lesssim_r 
\frac{\delta\sqrt{m}}{\sigma_\mathrm{min}(D^{-r}E)} , 
\end{equation}
where $D$ is the $m\times m$ difference matrix defined by
\begin{equation}\label{def-D}
D_{ij} := \left \{ 
\begin{array}{rl}
1, & \mbox{if } i=j, \cr
-1, & \mbox{if } i=j+1, \cr
0, & \mbox{otherwise,}
\end{array}
\right.
\end{equation}
and $\sigma_{\min}(D^{-r}E)$ stands for the smallest singular value
of $D^{-r}E$.

\subsection*{Contributions}
For the compressed sensing application that is the subject of this
paper, $E$ will simply be a sub-matrix of the measurement matrix
$\Phi$, hence it may have been found by sampling an i.i.d.~random
variable.  Minimum singular values of random matrices with
i.i.d.~entries have been studied extensively in the mathematical
literature.  For an $m \times k$ random matrix $E$ with i.i.d.~entries
sampled from a sub-Gaussian distribution with zero mean and unit
variance,\footnote{ As mentioned earlier, 
we do not normalize the measurement matrix $\Phi$
in the quantization setting.}  one has
\begin{equation}
\sigma_{\min}(E) \geq \sqrt{m}-\sqrt{k}
\end{equation}
with high probability \cite{RV2009}. Note that in general $D^{-r}E$
would not have i.i.d.~entries.  A naive lower bound for
$\sigma_{\min}(D^{-r}E)$ would be
$\sigma_{\min}(D^{-r})\sigma_{\min}(E)$.  However (see Proposition
\ref{sing_val_Dr}), $\sigma_{\min}(D^{-r})$ satisfies
\begin{equation}
\sigma_{\min}(D^{-r}) \asymp_r 1,
\end{equation}
and therefore this naive product bound yields no improvement on the
reconstruction error for $\Sigma\Delta$-quantized measurements over
the bound \eqref{err_bound_PCM_alt} for PCM-quantized ones.  In fact,
the true behavior of $\sigma_{\min}(D^{-r}E)$ turns out to be
drastically different and is described in Theorem~\ref{main_thm_1},
one of our main results (see also Theorem \ref{bound_sigma_min}).

For simplicity, we shall work with standard i.i.d.~Gaussian variables
for the entries of $E$. In analogy with our earlier notation, we
define the ``oversampling ratio'' $\lambda$ of the frame $E$ by
\begin{equation}
\lambda := \frac{m}{k}.
\end{equation}
\begin{bigthm}\label{main_thm_1}
Let $E$ be an $m\times k$ random matrix whose entries are
i.i.d. $\G(0,1)$. For any $\alpha \in (0,1)$, if 
$\lambda \geq c (\log m)^{1/(1-\alpha)}$, then with probability at
least $1 - \exp(-c' m \lambda^{-\alpha})$,
\begin{equation}\label{our-sing-val-bound}
\sigma_{\min}(D^{-r}E) \gtrsim_r \lambda^{\alpha(r-\frac{1}{2})}\sqrt{m},
\end{equation}
which yields the reconstruction error bound
\begin{equation} \label{eq14}
\|x - \hat x_{\mathrm{\Sigma\Delta}} \|_2 \lesssim_r  \lambda^{-\alpha(r-\frac{1}{2})} \delta.
\end{equation}
\end{bigthm}

While the kind of decay in 
this error bound is familiar to $\Sigma\Delta$ modulation, the 
domain of applicability of this result is rather 
surprising. Previously, the only
setting in which this type of approximation accuracy could be achieved 
(with or without Sobolev duals)
was the case of highly structured frames (e.g. when
the frame vectors are found by sampling along 
a piecewise smooth frame path). 
Theorem~\ref{main_thm_1} shows that such an
accuracy is obtained even when the analysis frame is a random Gaussian
matrix, provided the reconstruction is done via Sobolev duals.

\ignore{This lower bound for the smallest
singular value of $D^{-r}E$ is bigger than the smallest singular value of
$E$ by a factor of $\lambda^{\alpha(r-\frac{1}{2})}$. 
It is perhaps interesting to note that this factor is 
roughly about the same size as the 
$k$th {\em largest} singular value of $D^{-r}$, i.e.,
\begin{equation}
\sigma_{\min}(D^{-r}E) \gtrsim_r \sigma_k(D^{-r}) \sigma_{\min}(E).
\end{equation}
}

In the compressed sensing setting, one needs
\eqref{our-sing-val-bound} to be uniform for all the frames $E$ that
are found by selecting $k$ columns of $\Phi$ at a time. 
The proof of Theorem \ref{main_thm_1} extends
in a straightforward manner using a
standard ``union bound'' argument, provided $\lambda$ is known to be
slightly larger. More precisely, if $\Phi$ is an $m \times N$ matrix
whose entries are i.i.d.~according to $\G(0,1)$, and if $\lambda:=m/k
\geq c (\log N)^{1/(1-\alpha)}$, then \eqref{our-sing-val-bound} holds
for all $E = \Phi_T$ with $\#T \leq k$ with the same type of
probability bound (with new constants). This result can be utilized to
improve the reconstruction accuracy of a sparse signal $x$ from its
$\sd$-quantized compressed sensing measurements if the support $T$ of
$x$ is known. This is because if $T$ is known, $\Phi_T$ is known, and
its Sobolev dual can be found and used in the reconstruction. On the
other hand, for most signals, recovering the exact or approximate
support is already nearly guaranteed by the robust recovery result
shown in \eqref{ell1_eps_prog_alt} together with the stability of the
associated $\sd$ quantizer. For example, a simple sufficient condition
for full recovery of the support is that all the $|x_j|$ for $j \in T$
be larger than $C\|y - q_{\mathrm{\Sigma\Delta}}\|_2$ for a suitable
constant $C$. A precise version of this condition is stated in
Theorem~\ref{main_thm_2}.

In light of all these results, we propose $\Sigma\Delta$ quantization
as a more effective alternative of PCM (independent quantization) for
compressed sensing. With high probability on the measurement matrix, a
significant improvement of the reconstruction accuracy of sparse
signals can be achieved through a two-stage recovery procedure:
\begin{enumerate}
\item {\bf Coarse recovery:} $\ell_1$-minimization (or any other robust
recovery procedure) applied to $q_{\mathrm{\Sigma\Delta}}$
yields an initial, ``coarse'' approximation $x^\#$ of $x$, 
and in particular, the exact (or approximate) support $T$ of $x$.
\item {\bf Fine recovery:} 
Sobolev dual of the frame $\Phi_T$ applied to $q_{\mathrm{\Sigma\Delta}}$
yields a finer approximation $\hat x_{\mathrm{\Sigma\Delta}}$ of $x$.
\end{enumerate}

Combining all these, our second main theorem follows (also 
see Theorem~\ref{main_thm_2_1}):

\begin{bigthm}\label{main_thm_2} Let $\Phi$ be an $m \times N$ matrix whose entries are
  i.i.d.~according to $\G(0,1)$. Suppose $\alpha \in (0,1)$ and
  $\lambda:=m/k \geq c (\log N)^{1/(1-\alpha)}$ where
  $c=c(r,\alpha)$. Then
  there are two constants $c'$ and $C$ that depend only on $r$ such that 
  with probability at least $1 - \exp(-c' m \lambda^{-\alpha})$ on the 
  draw of $\Phi$, the
  following holds: For every $x\in \Sigma^N_k$
  such that $\min_{j\in \mathrm{supp}(x)} |x_j| \ge C \delta$,
  the reconstruction $\hat{x}_{\sd}$ satisfies
\begin{equation}\label{sd_error_main}
\|x - \hat{x}_{\sd}\|_2 \lesssim_r  \lambda^{-\alpha(r-\frac{1}{2})}
\delta.
\end{equation} 
\end{bigthm}


To put this result in perspective, note that the approximation error
given in \eqref{sd_error_main} decays as the ``redundancy''
$\lambda=\frac{m}{k}$ increases. In fact, by using an arbitrarily
high order
$\sd$ scheme, we can make this decay faster than any power law 
(albeit with higher constants). Note that such a decay is
not observed in the reconstruction error bound for PCM given
in \eqref{err_bound_PCM_alt}. Of course, one could argue that these
upper bounds may not reflect the actual behavior of the
error. However, in the setting of frame quantization the performance
of PCM is well investigated. In particular, let $E$ be an $m\times k$
real matrix, and let $K$ be a bounded set in $\R^k$. For $x\in K$,
suppose we obtain $q_{\text{PCM}}(x)$ by quantizing the entries
of $y=Ex$ using PCM with alphabet $\A=\delta \Z$. Let
$\Delta_{\text{opt}}$ be an optimal decoder.
Then, Goyal et al. show in \cite{GVT} that 
$$
\left[\mathbb{E} \left\|x-\Delta_{\text{opt}}(q_{\text{PCM}}(x))\right\|^2_2\right]^{1/2}
\gtrsim \lambda^{-1} \delta
$$
where $\lambda=m/k$ and the expectation is with respect a probability
measure on $x$ that is, for example, absolutely continuous.  This lower bound
limits the extent to which one can improve the reconstruction by means
of alternative reconstruction algorithms from PCM-quantized compressed
sensing measurements. On the other hand, setting, for example,
$\alpha=3/4$ in Theorem~\ref{main_thm_2} we observe that if we use a
second-order $\sd$ scheme to quantize the measurements, and if we
adopt the two-stage recovery procedure proposed above, the resulting
approximation will be superior to that produced optimally from
PCM-quantized measurements, provided $m/k$ is sufficiently large.

It is possible to imagine more sophisticated and more effective quantization 
and recovery algorithms for compressed sensing. However
using $\Sigma\Delta$ quantization has a number of appealing features:

\begin{itemize}
\item It produces {\bf more accurate} approximations than any known
  quantization scheme in this setting (even when sophisticated
  recovery algorithms are employed).
\item It is {\bf modular} in the sense that if the fine recovery stage is not available or practical to implement, then the standard (coarse) recovery procedure can still be applied as is. 
\item It is {\bf progressive} in the sense that if new measurements arrive
(in any given order), noise shaping can be continued on these measurements 
as long as the state of the system ($r$ real values for an $r$th order scheme)
has been stored. 
\item It is {\bf universal} in the sense that
it uses no information about the measurement matrix or the signal.
\end{itemize}

The paper is organized as follows.
We review the basics of $\Sigma\Delta$ quantization
and Sobolev duals in frame theory in Section \ref{sec2}, followed by 
the reconstruction 
error bounds for random Gaussian frames in Section \ref{main_proof}.
We then present the 
specifics of our proposed quantization and recovery algorithm for 
compressed sensing in Section \ref{SD_CS}. We present our
numerical experiments in Section \ref{numerics} and conclude with extensions
to more general settings in Section \ref{extensions}.

\section{Background on $\Sigma\Delta$ quantization of frame expansions}
\label{sec2}

\subsection*{$\Sigma\Delta$ quantization}

The governing 
equation of a standard $r$th order $\Sigma\Delta$ quantization scheme 
with input $y = (y_j)$ and output $q = (q_j)$ is
\begin{equation}\label{r-th-SD-eq}
(\Delta^ru)_j = y_j - q_j,~~~j=1,2,\dots,
\end{equation}
where the $q_j \in \A$ are chosen according to some quantization rule
given by
\begin{equation}
q_j = Q(u_{j-1},\dots,u_{j-T},y_j,\dots,y_{j-S}).
\end{equation}
Not all $\Sigma\Delta$ quantization schemes are presented (or
implemented) in this canonical form, but they all can be rewritten as
such for an appropriate choice of $r$ and $u$. We shall not be
concerned with the specifics of the mapping $Q$, except that we need
$u$ to be bounded. The smaller the size of the alphabet $\A$ gets
relative to $r$, the harder it is to guarantee this property. The
extreme case is $1$-bit quantization, i.e., $|\A|=2$, which is
typically the most challenging setting. We will not be working in this
case. In fact, for our purposes, $\A$ will in general have to be
sufficiently fine to allow for the recovery of the support of sparse
signals. In order to avoid technical difficulties, we shall work with
the infinite alphabet $\A = \delta \Z$, but also note that only a
finite portion of this alphabet will be used for bounded signals.  A
standard quantization rule that has this ``boundedness'' property is
given by the greedy rule which minimizes $|u_j|$ given
$u_{j-1},\dots,u_{j-r}$ and $y_j$, i.e.,
\begin{equation}\label{greedy-quant}
q_j = \arg\min_{a \in \A} \Big |\sum_{i=1}^r (-1)^{i-1} 
{r \choose i} u_{j-i} + y_j -  a \Big|.
\end{equation}
It is easy to check that with this rule, one has $|u_j| \leq 2^{-1} \delta$
and $|y_j-q_j| \leq 2^{r-1} \delta$. In turn, if $\|y\|_\infty < C$,
then one needs only $L:=2 \lceil \frac{C}{\delta}\rceil +2^r + 1$
levels. In this case, the associated quantizer is said to be $\log_2
L$-bit, and we have
\begin{equation}
\|u\|_\infty \lesssim \delta \mbox{ and } \|y-q\|_\infty \lesssim_r \delta.
\end{equation}
With more stringent quantization rules, the first inequality would also
have an $r$-dependent constant. In fact, it is known that 
for quantization rules with a $1$-bit alphabet, 
this constant will be as large as $O(r^r)$, e.g., see \cite{DD,G-exp}.
In this paper, unless otherwise stated, we shall be working with the 
greedy quantization rule of \eqref{greedy-quant}.

The initial condition of the recursion in \eqref{r-th-SD-eq} can be
set arbitrarily, but it will be convenient for us to set them equal to
zero for finite frames.  With $u_{-r+1} = \cdots = u_0 = 0$, and
$j=1,\dots,m$, the difference equation \eqref{r-th-SD-eq} can be
rewritten as a matrix equation
\begin{equation}\label{r-th-SD-matrix}
D^r  u = y - q,
\end{equation}
where $D$ is as in \eqref{def-D}.
\ignore{
More specifically, we have 
\begin{equation}
D = \begin{bmatrix} 
1 & 0     & 0      & 0      & \cdots & 0 \\
-1 & 1      & 0     & 0      & \cdots & 0\\ 
0 & -1      & 1     & 0      & \cdots & 0\\ 
 & \ddots & \ddots & \ddots &   \ddots     &  \\ 
0  &		0		 &  \cdots      &    -1    &     1   &  0 \\
0  &		0		 &  \cdots      &    0    &     -1   &  1
\end{bmatrix}
\end{equation}
}

As before, we assume $E$ is an $m\times k$ matrix whose rows form the
analysis frame and $F$ is a $k\times m$ left inverse of $E$
whose columns form the dual (synthesis) frame. 
Given any $x \in \R^k$, we set $y = Ex$, and define its
$r$th order $\Sigma\Delta$ quantization $q_{\mathrm{\Sigma\Delta}}$
and its reconstruction $\hat x_{\mathrm{\Sigma\Delta}} := 
Fq_{\mathrm{\Sigma\Delta}}$. 
Substituting \eqref{r-th-SD-matrix} into \eqref{SD-err-1},
we obtain the error expression
\begin{equation}\label{SD-err}
x - \hat x = FD^r u.
\end{equation}
With this expression, $\|x -\hat x\|$ can be bounded for any norm
$\| \cdot \|$ simply as 
\begin{equation}\label{sup-err-bound}
\|x - \hat x \| \leq \|u\|_\infty \sum_{j=1}^m \| (FD^r)_j \|.
\end{equation}
Here $(FD^r)_j$ is the $j$th column of $FD^r$. This bound is also valid
in infinite dimensions, and in fact has been used extensively
in the mathematical treatment of oversampled A/D conversion of bandlimited
functions.
\ignore{For example, for $\| \cdot \| = \| \cdot \|_\infty$ and 
$f_j = \tau \varphi(\cdot - j\tau)$, $j \in \Z$,
one can bound the (infinite) sum on the right hand side of 
\eqref{sup-err-bound} by 
$\tau^{r} \|\varphi^{(r)} \|_1$.}

For $r=1$, and the $\ell_2$ norm, the sum term on the right hand side
motivated the study of the so-called {\em frame variation} defined by
\begin{equation}
V(F) := \sum_{j=1}^{m} \|f_j-f_{j+1}\|_2,
\end{equation}
where $(f_j)$ are the columns of $F$, and one defines $f_{m+1} = 0$.
Higher-order frame variations 
to be used with higher-order $\Sigma\Delta$ schemes are defined
similarly, see \cite{BPY,BPY2}.
Frames (analysis as well as synthesis) 
that are obtained via uniform sampling a smooth curve in $\R^k$
(so-called {\em frame path}) are typical in many settings. However,
the ``frame variation bound'' is useful in finite dimensions when
the frame path terminates smoothly. Otherwise, it does not provide
higher-order reconstruction accuracy. Designing smoothly
terminating frames can be technically challenging, e.g., \cite{BPA}.

\subsection*{Sobolev duals}
Recently, a more straightforward approach was proposed in \cite{LPY}
for the design of (alternate) duals of finite frames for
$\Sigma\Delta$ quantization.  Here, one instead considers the operator
norm of $FD^r$ on $\ell_2$ and the corresponding bound
\begin{equation}\label{SD-op-bound}
\|x - \hat x \|_2 \leq \|FD^r \|_\mathrm{op} \|u\|_2.
\end{equation}
Note that this bound is not available in the infinite dimensional setting of 
bandlimited functions due to the fact that $u$ is typically not in $\ell_2$. 
It is now natural to minimize
$\|FD^r \|_\mathrm{op}$ over all dual frames of a given analysis frame $E$.
These frames, introduced in \cite{BLPY}, have been called Sobolev duals, in analogy with $\ell_2$-type
Sobolev (semi)norms.

$\Sigma\Delta$ quantization algorithms are normally designed 
for analog circuit operation, so they
control $\|u\|_\infty$, which would control $\|u\|_2$ only in a
suboptimal way. However, it turns out that there are important
advantages in working with the $\ell_2$ norm in the analysis.
The first advantage is that
Sobolev duals are readily available by an explicit formula. The
solution $F_{\mathrm{sob},r}$ of the optimization problem
\begin{equation}\label{opt-frame-r}
\min_{F} \|FD^r\|_{\mathrm{op}} \mbox{ subject to } FE = I
\end{equation}
is given by the matrix equation
\begin{equation}\label{sobolev-formula}
F_{\mathrm{sob},r} D^r = (D^{-r}E)^\dagger,
\end{equation}
where $^\dagger$ stands for the Moore-Penrose inversion operator,
which, in our case, is given
by $E^\dagger := (E^* E)^{-1}E^*$. Note that for $r=0$ (i.e., no
noise-shaping, or PCM), one simply obtains $F = E^\dagger$, the
canonical dual frame of $E$.

The second advantage of this approach is that highly developed methods
are present for spectral norms of matrices, especially in the random
setting.  Plugging \eqref{sobolev-formula} into \eqref{SD-op-bound},
it immediately follows that
\begin{equation}\label{err-bound-Sob}
\|x - \hat x\|_2 \leq \|(D^{-r}E)^\dagger\|_{\mathrm{op}} \|u\|_2
= \frac{1}{\sigma_{\min}(D^{-r}E)}\|u\|_2,
\end{equation}
where $\sigma_{\min}(D^{-r}E)$ stands for the smallest singular value
of $D^{-r}E$.

\section{Reconstruction error bound for random frames}
\label{main_proof}

In what follows, $\sigma_j(A)$ will denote the $j$th largest singular
value of the matrix $A$. Similarly, $\lambda_j(B)$ will denote the 
$j$th largest eigenvalue of the Hermitian matrix $B$. Hence, we have
$\sigma_j(A) = \sqrt{\lambda_j(A^*A)}$. We will also use the notation
$\Sigma(A)$ for the diagonal matrix of singular values of $A$, with 
the convention $(\Sigma(A))_{jj} = \sigma_j(A)$. All matrices
in our discussion will be real valued and the Hermitian conjugate
reduces to the transpose.

We have seen that the main object of interest for the reconstruction error 
bound is $\sigma_{\min}(D^{-r}E)$ for a random frame $E$. 
Let $H$ be a square matrix.
The first observation we make is that when $E$ is
i.i.d.~Gaussian, the distribution of $\Sigma(HE)$
is the same as the distribution of $\Sigma(\Sigma(H)E)$. To see this,
let $U \Sigma(H) V^*$ be the singular value decomposition of 
$H$ where $U$ and $V$ are unitary matrices. Then $HE = U \Sigma(H) V^* E$. 
Since the unitary transformation $U$ does not alter singular values, 
we have $\Sigma(HE) = \Sigma(\Sigma(H) V^* E)$, and because of the unitary 
invariance of the i.i.d.~Gaussian measure, the matrix $\tilde E
:= V^*E$ has the same distribution as $E$, hence the claim. 
Therefore it suffices to study the singular values of $\Sigma(H)E$.
In our case, $H = D^{-r}$ and we first need information
on the deterministic object $\Sigma(D^{-r})$.
The following result will be sufficient for our purposes:

\begin{prop}\label{sing_val_Dr}
Let $r$ be any positive integer and $D$ be as in \eqref{def-D}. There are
positive numerical constants $c_1(r)$ and $c_2(r)$, independent of $m$, 
such that 
\begin{equation}\label{sing_val_Dr_bounds}
c_1(r) \Big(\frac{m}{j}\Big)^r \leq \sigma_j(D^{-r}) \leq 
c_2(r) \Big(\frac{m}{j}\Big)^r, ~~j=1,\dots,m.
\end{equation}
\end{prop}

The proof of this result is rather standard in the study of 
Toeplitz matrices, and is 
given in Appendix \ref{sec_sing_val_D}.

\subsection{Lower bound for $\sigma_{\min}(D^{-r}E)$}
\label{lower_bound}

In light of the above discussion, 
the distribution of $\sigma_{\min}(D^{-r}E)$
is the same as that of 
\begin{equation}
\inf_{\|x\|_2 = 1} \|\Sigma(D^{-r})E x\|_2.
\end{equation} 
We replace $\Sigma(D^{-r})$ with an arbitrary diagonal matrix
$S$ with $S_{jj} =: s_j > 0$. The first two results will concern upper
bounds for the norm of independent but 
non-identically distributed Gaussian vectors. They are rather standard,
but we include them for the definiteness of our discussion when they will
be used later.

\begin{prop} \label{power2} 
Let $\xi \sim \G(0,\frac{1}{m}\mathrm{I}_m)$.
For any $\Theta  > 1$, 
\begin{equation}
\mathbb{P} \left( \sum_{j=1}^m s^2_j \xi_j^2 > \Theta  \|s\|^2_\infty \right)
\leq \Theta ^{m/2} e^{-(\Theta -1)m/2}.
\end{equation}
\end{prop}

\begin{proof} 
Since $s_j  \leq \|s\|_\infty$ for all $j$, we have 
\begin{equation}
\mathbb{P} \left( \sum_{j=1}^m s^2_j \xi_j^2 > \Theta  \|s\|^2_\infty \right)
\leq 
\mathbb{P} \left( \sum_{j=1}^m \xi_j^2 > \Theta  \right).
\end{equation}
This bound is the (standard) 
Gaussian measure of the complement of a sphere of radius 
$\sqrt{m\Theta}$ and can be estimated very accurately. We use a simple
approach via
\begin{eqnarray}
\mathbb{P} \left( \sum_{j=1}^m \xi_j^2 > \Theta  \right)
& \leq & 
\min_{\lambda \geq 0}~ \int_{\mathbb{R}^m} e^{-\left(\Theta  
-\sum_{j=1}^m x_j^2\right)\lambda/2} \prod_{j=1}^m 
e^{-mx_j^2/2}\,\frac{\mathrm{d}x_j}{\sqrt{2\pi/m}} \cr
& = &
\min_{\lambda \geq 0}~ e^{-\lambda \Theta /2}
(1-\lambda /m)^{-m/2} \cr
& = & 
\Theta ^{m/2} e^{-(\Theta -1)m/2},
\end{eqnarray}
where in the last step we set $\lambda = m(1-\Theta ^{-1})$.
\end{proof}

\begin{lemma}\label{lower}
Let $E$ be an $m \times k$ 
random matrix whose entries are i.i.d.~$\G(0,\frac{1}{m})$.
For any $\Theta  > 1$, consider the event
$$\mathcal{E} := \left\{\|S E \|_{\ell^k_2\to \ell^m_2} \leq 2  
\sqrt{\Theta} \|s\|_\infty \right \}.$$
Then
$$\mathbb{P} \left( \mathcal{E}^c \right )
\leq 5^k \Theta ^{m/2} e^{-(\Theta -1)m/2}. $$
\end{lemma}

\begin{proof}
We follow the same approach as in \cite{JL_RIP}.
The maximum number of $\rho$-distinguishable points on the
unit sphere in $\R^k$ is at most $(\frac{2}{\rho}+1)^k$.
(This follows by a volume argument\footnote{Balls with radii 
$\rho/2$ and centers
at a $\rho$-distinguishable set of points on the unit sphere are mutually
 disjoint and 
are all contained in the ball of radius $1+\rho/2$ centered at
the origin. Hence there can be at most $(1+\rho/2)^k/(\rho/2)^k$ of them.} 
as in e.g., \cite[p.487]{LorentzCA2}.) 
Fix a maximal set $Q$ of $\frac{1}{2}$-distinguishable points 
of the unit sphere in $\mathbb{R}^k$ with $\# Q \leq 5^k$.  
Since $Q$ is maximal, it is a $\frac{1}{2}$-net for the unit sphere.
For each $q \in Q$, consider $\xi_j = (Eq)_j$,
$j = 1,\dots,m$. Then $\xi\sim \G(0,\frac{1}{m}\mathrm{I}_m)$. 
As before, we have
$$ \|S E q\|_2^2 = \sum_{j=1}^m s^2_j \xi_j^2. $$
Let $\mathcal{E}(Q)$ be the event 
$\left \{ \|S E q \|_2 \leq \sqrt{\Theta}  \|s\|_\infty, ~~
\forall q \in Q \right \}$. Then, by Proposition \ref{power2}, we have
the union bound
\begin{equation}
\mathbb{P} \left( \mathcal{E}(Q)^c \right)
\leq 5^k \Theta ^{m/2} e^{-(\Theta -1)m/2}.
\end{equation}
Assume the event $\mathcal{E}(Q)$, and let 
$M =  \|S E \|_{\ell^k_2\to \ell^m_2}$. For each $\|x\|_2 = 1$, there is $q \in Q$ with $\|q-x\|_2 \leq 1/2$,
hence  
$$ \|SEx\|_2 \leq \|SEq\|_2 + \|SE(x-q)\|_2 \leq  
\sqrt{\Theta} \|s\|_\infty + \frac{M}{2}.$$
Taking the supremum over all $x$ on the unit sphere, we obtain
$$ M \leq \sqrt{\Theta} \|s\|_\infty + \frac{M}{2}, $$
i.e., $\|S E \|_{\ell^k_2\to \ell^m_2} \leq 2\sqrt{\Theta} \|s\|_\infty$.
Therefore $\mathcal{E}(Q) \subset \mathcal{E}$, and the result follows.
\end{proof}

The following estimate concerns a lower bound for the Euclidean norm
of $(s_1\xi_1,\dots,s_m\xi_m)$. It is not sharp when 
the $s_j$ are identical, but it will be useful for our problem
where $s_j = \sigma_j(D^{-r})$ obey a power law (see Corollary
\ref{power1}).

\begin{prop} \label{crude1} 
Let $\xi \sim \G(0,\frac{1}{m}\mathrm{I}_m)$.
For any $\gamma > 0$, 
\begin{equation}
\mathbb{P} \left( \sum_{j=1}^m s^2_j \xi_j^2 < \gamma \right)
\leq \min_{1 \leq L \leq m}
\left( \frac{e\gamma m}{L}\right)^{L/2} (s_1 s_2\cdots s_L)^{-1}.
\end{equation}
\end{prop}

\begin{proof}
For any $t \geq 0$ and any integer $L \in \{1,\dots,m\}$, we have
\begin{eqnarray}
\mathbb{P} \left( \sum_{j=1}^m s^2_j \xi_j^2 < \gamma \right)
& \leq &  
\int_{\mathbb{R}^m} e^{\left(\gamma - 
\sum_{j=1}^m s^2_j x_j^2\right)t/2 } \prod_{j=1}^m 
e^{-mx_j^2/2}\,\frac{\mathrm{d}x_j}{\sqrt{2\pi/m}}
\cr
& = & e^{t\gamma/2} 
\prod_{j=1}^m \int_{\mathbb{R}} 
e^{-x_j^2(m+t s^2_j)/2} \,\frac{\mathrm{d}x_j}{\sqrt{2\pi/m}} \cr 
& = & e^{t \gamma/2}
\prod_{j=1}^m (1+t s^2_j/m)^{-1/2} \cr
& \leq & e^{t \gamma/2}
\prod_{j=1}^L (t s^2_j/m)^{-1/2} \cr
& \leq &
e^{t \gamma/2} (m/t)^{L/2} (s_1 s_2 \cdots s_L)^{-1}.
\end{eqnarray}
For any $L$, we can set $t = L/\gamma$, which is the critical point of
the function $t \mapsto e^{t \gamma} t^{-L}$. Since
$L$ is arbitrary, the result follows.
\end{proof}

\begin{cor} \label{power1} Let 
$\xi \sim \G(0,\frac{1}{m}\mathrm{I}_m)$, $r$ be a positive integer, 
and $c_1 > 0$ be such that 
\begin{equation}\label{hypo-s-r}
s_j \geq c_1 \left(\frac{m}{j} \right)^{r},~~~~j=1,\dots,m.
\end{equation}
Then for any $\Lambda  \geq 1$ and $m \geq \Lambda $, 
\begin{equation}
\mathbb{P} \left( \sum_{j=1}^m s^2_j \xi_j^2 < c_1^2 \Lambda ^{2r-1} \right)
< (60 m/\Lambda )^{r/2} e^{-m (r{-}1/2)/\Lambda }.
\end{equation}
\end{cor}

\begin{proof} 
By rescaling $s_j$, we can assume $c_1=1$.
For any $L \in \{1,\dots,m\}$, we have
$$ (s_1 s_2\cdots s_L)^{-1} \leq \frac{(L!)^r}{m^{rL}} 
< (8 L)^{r/2} \left( \frac{L^r}{e^rm^r} \right)^L,$$
where we have used the coarse estimate $L! < e^{1/12L}(2\pi L)^{1/2}(L/e)^{L}
< (8L)^{1/2}(L/e)^L$.
Setting $\displaystyle \gamma = \Lambda ^{2r-1}$ in Proposition 
\ref{crude1}, we obtain
\begin{equation}
\mathbb{P} \left( \sum_{j=1}^m s^2_j \xi_j^2 < \Lambda ^{2r-1} \right)
<  (8 L)^{r/2} 
\left[  \left( \frac{\Lambda L}{em}\right)^{L} \right]^{r-1/2}. 
\label{est1}
\end{equation}
We 
set $L = \lfloor \frac{m}{\Lambda } \rfloor$. Since $1 \leq \Lambda  \leq m$, it is guaranteed that $1 \leq L \leq m$. 
Since $\Lambda L\leq m$, we get
$$ \left( \frac{\Lambda L}{em}\right)^{L} \leq e^{-L} < e^{1-\frac{m}{\Lambda }} $$
Plugging this in \eqref{est1} and using $8e^2 < 60$, we find
\begin{eqnarray}
\mathbb{P} \left( \sum_{j=1}^m s^2_j \xi_j^2 < \Lambda ^{2r-1} \right)
& < &  (60 m/\Lambda )^{r/2} e^{-m (r{-}1/2)/\Lambda }.
\end{eqnarray}
\end{proof}

\begin{thm} \label{upper}
Let $E$ be an $m \times k$
random matrix whose entries are i.i.d.~$\G(0,\frac{1}{m})$,
$r$ be a positive integer, 
and assume that the entries $s_j$ of the diagonal matrix $S$ satisfy
\begin{equation}\label{hypo-s-r-2}
c_1 \left(\frac{m}{j} \right)^{r} \leq s_j \leq c_2 m^r,~~~~j=1,\dots,m.
\end{equation}
Let $\Lambda  \geq 1$ be any number and assume 
$m \geq \Lambda $. Consider the event
$$ \mathcal{F} := \left \{\|S E x \|_2 \geq  
\frac{1}{2} c_1 \Lambda ^{r-1/2} \|x\|_2, ~\forall x \in \mathbb{R}^k \right \}.$$
Then 
$$\mathbb{P} \left( \mathcal{F}^c \right )
\leq
5^k e^{-m/2} + 8^r 
\left( 17 c_2/c_1\right)^k \Lambda ^{k/2} 
\left (\frac{m}{\Lambda } \right)^{r(k+1/2)}
e^{-m(r{-}1/2)\Lambda}. $$
\end{thm}

\begin{proof}
Consider a $\rho$-net $\tilde Q$
of the unit sphere of $\mathbb{R}^k$ with $\# \tilde Q \leq 
\big( \frac{2}{\rho}+1 \big )^k$ where the value of $\rho < 1$
will be chosen later. Let 
$\tilde{\mathcal{E}}(\tilde Q)$ be the event
$\left \{\|S E q \|_2 \geq  c_1 \Lambda ^{r-1/2}, ~~\forall 
q \in \tilde Q \right \}$. By Corollary \ref{power1}, we know that
\begin{equation}\label{prob-bound-union}
\mathbb{P} \left( \tilde{\mathcal{E}}(\tilde Q)^c \right)
\leq  \left( \frac{2}{\rho}+1 \right )^k \left( \frac{60m}{\Lambda } \right )^{r/2} 
e^{-m (r{-}1/2)/\Lambda }.
\end{equation}
Let $\mathcal{E}$ be the event in Lemma \ref{lower} with $\Theta = 4$.
Let $E$ be any given matrix in the event 
$\mathcal{E} \cap \tilde{\mathcal{E}}(\tilde Q)$.
For each $\|x\|_2 = 1$, there is $q \in \tilde Q$ with $\|q-x\|_2 \leq \rho$,
hence by Lemma \ref{lower}, we have
$$\|SE(x-q)\|_2 \leq 4 \|s\|_\infty \|x-q\|_2 \leq 4 c_2 m^r \rho. $$
Choose 
$$\rho = \frac{c_1 \Lambda^{r-1/2}}{8 c_2 m^r} = \frac{c_1}{8c_2\sqrt{\Lambda}}
\Big(\frac{\Lambda}{m}\Big)^r. $$
Hence 
$$ \|SEx\|_2 \geq \|SEq\|_2 - \|SE(x-q)\|_2 \geq  
c_1 \Lambda^{r-1/2} - 4 c_2 m^r\rho
= \frac{1}{2} c_1 \Lambda^{r-1/2}.
$$
This shows that $\mathcal{E} \cap \tilde{\mathcal{E}}(\tilde Q) 
\subset \mathcal{F}$.
Clearly, $\rho \leq 1/8$ by our choice of parameters
and hence $\frac{2}{\rho} +1 
\leq \frac{17}{8\rho}$.
Using the probability bounds of Lemma \ref{lower} and \eqref{prob-bound-union},
we have
\begin{eqnarray}
\mathbb{P} \left( \mathcal{F}^c \right)
& \leq & 5^k 4^{m/2} e^{-3m/2} +
 \left( \frac{17}{8\rho} \right )^k \left( \frac{60m}{\Lambda } \right )^{r/2} 
e^{-m(r{-}1/2)/\Lambda} \cr
& \leq & 5^k e^{-m/2} + 8^r 
(17c_2/c_1)^k \Lambda ^{k/2} 
\left (\frac{m}{\Lambda } \right)^{r(k+1/2)}
e^{-m(r{-}1/2)/\Lambda},
\end{eqnarray}
where we have used $2 < e$ and $\sqrt{60} < 8$ for simplification.
\end{proof}

The following theorem is now a direct corollary of the above
estimate. 

\begin{thm} \label{bound_sigma_min}
Let $E$ be an $m \times k$
random matrix whose entries are i.i.d.~$\G(0,\frac{1}{m})$,
$r$ be a positive integer, $D$ be the difference matrix defined in
\eqref{def-D}, and the constant $c_1 = c_1(r)$ be
as in Proposition \ref{sing_val_Dr}.
Let $0 < \alpha < 1$ be any number. Assume that
\begin{equation}
\lambda := \frac{m}{k} \geq c_3 (\log m)^{1/(1-\alpha)},
\end{equation}
where $c_3= c_3(r)$ is an appropriate constant. Then 
\begin{equation}
\mathbb{P} \left ( \sigma_{\mathrm{min}}(D^{-r}E) \geq c_1 
\lambda^{\alpha(r-1/2)} \right ) \geq 1 - 2e^{-c_4 m^{1-\alpha}k^\alpha}
\end{equation}
for some constant $c_4 = c_4(r) > 0$.
\end{thm}

\begin{proof}
Set $\Lambda  = \lambda ^\alpha$ in Lemma \ref{upper}.
We only need to show that 
$$
\max \left [
5^k e^{-m/2}, 8^r (17c_2/c_1)^k \Lambda ^{k/2} 
\left (\frac{m}{\Lambda } \right)^{r(k+1/2)}
e^{-m(r{-}1/2)/\Lambda}
\right ]
\leq e^{-c_4m^{1-\alpha}k^\alpha}.
$$
It suffices to show that 
$$ k \log 5 - m/2 \leq - c_4 m^{1-\alpha}k^\alpha $$
and 
$$r \log 8 + k \log(17c_2/c_1) +  \frac{1}{2}k \log \Lambda 
+ r(k+\frac{1}{2}) \log(m/\Lambda ) - (r{-}\frac{1}{2})\frac{m}{\Lambda }
\leq - c_4 m^{1-\alpha}k^\alpha.
$$
The first inequality is easily seen to hold if 
$\lambda \geq \frac{\log 5}{\frac{1}{2} - c_4}$.
For the second inequality, first
notice that $m/\Lambda  = m^{1-\alpha}k^\alpha$.
Since $k+1/2 \asymp k$, and $r - 1/2 \asymp r$,
it is easily seen that we only need to check that
$$ k \log m \leq c_5 \frac{m}{\Lambda }$$
for a sufficiently small $c_5$. This follows from our assumption on
$\lambda$ by setting $c_5 = 1/c_3^{1-\alpha}$.
\end{proof}

\begin{rem} By replacing $E$ in Theorem~\ref{bound_sigma_min} with $\sqrt{m} E$, we
obtain Theorem~\ref{main_thm_1}.
\end{rem}

\subsection{Implication for compressed sensing matrices}

\begin{thm}\label{sing_val_for_CS}
Let $r$, $D$, $c_1(r)$ be as in Theorem \ref{bound_sigma_min} and
$\Phi$ be an $m \times N$
random matrix whose entries are i.i.d.~$\G(0,\frac{1}{m})$.
Let $0 < \alpha < 1$ be any number and assume that
\begin{equation}
\lambda := \frac{m}{k} \geq c_6 (\log N)^{1/(1-\alpha)},
\end{equation}
where $c_6 = c_6(r)$ is an appropriate constant. Then with probability
at least $1 - 2e^{-c_7 m \lambda^{-\alpha}}$
for some $c_7 = c_7(r) > 0$, every 
$m \times k$ submatrix $E$ of $\Phi$ satisfies
\begin{equation} \label{eq_cs_sd}
\sigma_{\mathrm{min}}(D^{-r}E) \geq c_1 
\lambda^{\alpha(r-1/2)}.
\end{equation}
\end{thm}

\begin{proof}
We will choose $c_7 = c_4/2$, where $c_4$ is as in 
Theorem \ref{bound_sigma_min}. The proof will follow immediately by
a union bound once we show that 
$${N \choose k} \leq e^{\frac{1}{2}c_4 m^{1-\alpha}k^\alpha}.$$
Since ${N \choose k} \leq N^k$, it suffices to show that 
$$ k \log N \leq  \frac{c_4}{2} m^{1-\alpha}k^\alpha. $$
Both this condition and the hypothesis of Theorem \ref{bound_sigma_min} will
be satisfied if we choose 
$$c_6 = \max(c_3,(2/c_4)^{1/(1-\alpha)}).$$
\end{proof}

\begin{rem}
  If $\Phi$ is a Gaussian matrix with entries i.i.d. $\G(0,1)$ rather
  than $\G(0,\frac{1}{m})$, Theorem~\ref{sing_val_for_CS} applied to
  $\frac{1}{\sqrt{m}}\Phi$ implies that every $m\times k$ submatrix $E$ of $\Phi$
  satisfies
\begin{equation} \label{eq_cs_sd_1}
\sigma_{\mathrm{min}}(D^{-r}E) \geq c_1 
\lambda^{\alpha(r-1/2)}\sqrt{m}.
\end{equation}
\end{rem}

\section{$\Sigma\Delta$ quantization of compressed sensing
  measurements}
\label{SD_CS}

In this section we will assume that the conditions of Theorem
\ref{sing_val_for_CS} are satisfied for some $0 < \alpha < 1$ and $r$,
and the measurement matrix $\Phi$ that is drawn from $\G(0,1)$
yields \eqref{eq_cs_sd_1}. For definiteness, we also
assume that $\Phi$ admits the robust recovery constant $C_1=10$, i.e.,
the solution $x^\#$ of the program \eqref{ell1_eps_prog} satisfies
$$\|\hat{y}-y\|_2\le \epsilon \ \implies \ \|x-x^\#\|_2 \leq 10 \frac{1}{\sqrt{m}}\epsilon.
$$

Note again that our choice of normalization for the measurement matrix
$\Phi$ is different from the compressed sensing convention. As
mentioned in the Introduction, it is more appropriate to work with a
measurement matrix $\Phi \sim \G(0,1)$ in order to be able to use a
quantizer alphabet that does not depend on $m$. For this reason, in
the remainder of the paper, $\Phi$ shall denote an $m \times N$ matrix
whose entries are i.i.d. from $\G(0,1)$. 


Let $q :=
q_{\Sigma\Delta}$ be output of the standard greedy $r$th order
$\Sigma\Delta$ quantizer with the alphabet $\A = \delta \Z$ and input
$y$. As stated in Section \ref{sec2}, we know that $\|y - q\|_\infty
\leq 2^{r-1}\delta$ and therefore $\|y - q\|_2 \leq
2^{r-1}\delta\sqrt{m}$.
 
\subsection*{Coarse recovery and recovery of support}

Our first goal is to recover the support $T$ of $x$. For this purpose
we shall use a coarse approximation of $x$. Let 
\begin{equation}
x' := \arg \min \|z \|_1 \mbox{ subject to } 
\left \|\Phi z - q \right \|_2 \leq 
\epsilon :=  2^{r-1}\delta \sqrt{m}.
\end{equation}
By the robust recovery result (for our choice of normalization for
$\Phi$), we know that
$$\|x - x' \|_2 \leq \eta := 5\cdot 2^{r}\delta.$$

The simplest attempt to recover $T$ from $x'$ is to pick 
the positions of its $k$ largest entries. This attempt can fail
if some entry of $x_j$ on $T$ is smaller than $\eta$ for then
it is possible that $x'_j = 0$ and therefore $j$ is not picked. 
On the other hand, 
it is easy to see that if the smallest nonzero entry of $x$ is strictly
bigger than $2\eta$ in magnitude, then this method always succeeds.
(Since $\|x - x' \|_\infty \leq \eta$, the entries of
$x'$ are bigger than $\eta$ on $T$ and less than $\eta$ on $T^c$.) 
The constant $2$ can be replaced with $\sqrt{2}$ by a more careful
analysis, and can be pushed arbitrarily close to $1$ by picking more
than $k$ positions. The 
proposition below gives a precise condition on how well this can 
be done. We also provide a bound on how much of $x$ can potentially be
missed if no lower bound on $|x_j|$ is available for $j \in T$.

\begin{prop}\label{supp_recov}
Let $\|x - x'\|_{\ell_2^N} \leq \eta$, $T = \mathrm{supp}~x$
and $k = |T|$. For
any $k' \in \{k,\dots,N{-}1\}$, let $T'$ be the support of (any of) the 
$k'$ largest entries of $x'$. 
\begin{itemize}
\item[\rm{(i)}] $\|x_{T\setminus T'}\|_2 \leq \beta \eta$
where $\beta \leq \left(1+\frac{k}{k'}\right)^{1/2}$.
\item[\rm{(ii)}]
If $|x_j| > \gamma \eta$ for all $j \in T$, 
where $\gamma := \left (1 + \frac{1}{k'-k+1}\right )^{1/2}$,
then $T' \supset T$.
\end{itemize}
\end{prop}

\begin{proof}
\rm{(i)}
We have
\begin{equation}\label{error_split}
\sum_{j\in T} |x_j - x'_j|^2 + \sum_{j \in T^c} |x'_j|^2 =
\|x - x'\|_2^2 \leq \eta^2.
\end{equation}
In particular, this implies
\begin{equation}\label{TT'}
\sum_{j\in T\setminus T'} |x_j - x'_j|^2 + \sum_{j \in T'\setminus T} 
|x'_j|^2 \leq \eta^2.
\end{equation}

Suppose $T \setminus T' \not= \emptyset$. Then $T'\setminus T$ is also
nonempty. In fact, we have
$$|T'\setminus T| = |T \setminus T'| + k'-k.$$ 
Now, observe that
$$
\frac{1}{|T\setminus T'|} \sum_{j \in T\setminus T'} 
|x'_j|^2 \leq
\max_{j \in T\setminus T'} ~|x'_j|^2
\leq \min_{j \in T'\setminus T}~ |x'_j|^2
\leq \frac{1}{|T'\setminus T|} \sum_{j \in T'\setminus T} 
|x'_j|^2,
$$
which, together with \eqref{TT'} implies
$$ \|x_{T\setminus T'}\|_2 \leq \|x'_{T\setminus T'}\|_2
+ \|(x-x')_{T\setminus T'}\|_2
\leq 
\|x'_{T\setminus T'}\|_2 + 
\sqrt{\eta^2 - \frac{|T'\setminus T|}{|T\setminus T'|} 
\|x'_{T\setminus T'}\|^2_2}.
$$
It is easy to check that for any $A >0$, and any $0 \leq t \leq \eta/\sqrt{A}$,
\begin{equation}\label{lin_quad_max}
t + \sqrt{\eta^2 - At^2} \leq 
\left(1+\frac{1}{A}\right)^{1/2} \eta.
\end{equation}
The result follows by 
setting $A = |T'\setminus T|/|T\setminus T'|$ and noticing that 
$A \geq k'/k$.

\rm{(ii)}
Let $z_1 \geq \cdots \geq z_N$ be the decreasing 
rearrangement of $|x'_1|, \dots, |x'_N|$. We have
$$\sum_{j \in T} |x'_j|^2 \leq \sum_{i=1}^{k} z_i^2$$ 
so
$$\sum_{j \in T^c} |x'_j|^2 \geq \sum_{i=k+1}^{N} z_i^2 \geq
\sum_{i=k+1}^{k'+1} z_i^2 \geq (k'-k+1)z_{k'+1}^2.$$ 
Hence by \eqref{error_split} we have
$$\max_{j \in T} |x_j - x'_j|^2 + (k'-k+1) z_{k'+1}^2 \leq \eta^2.$$
Since $|x'_j| \geq |x_j| - |x_j - x'_j|$, the above inequality now implies
$$\min_{j \in T} |x'_j| 
\geq \min_{j \in T}|x_j| - \max_{j \in T}|x_j - x'_j|
\geq \min_{j \in T}|x_j| - \sqrt{\eta^2 - (k'-k+1)z_{k'+1}^2}.$$
Now, another application of \eqref{lin_quad_max} with 
$A = k'-k+1$ yields
$$- \sqrt{\eta^2 - (k'-k+1)z_{k'+1}^2} \geq z_{k'+1} - \gamma \eta$$
and therefore
$$\min_{j \in T} |x'_j| \geq \min_{j \in T}|x_j| + 
z_{k'+1} - \gamma \eta > z_{k'+1} = \max_{j \in T'^c} |x'_j|.$$
It is then clear that $T \subset T'$ because if $T'^c \cap T \not= \emptyset$, the
inequality
$$\max_{j \in T'^c} |x'_j| \geq \max_{j \in T'^c \cap T} |x'_j|
\geq \min_{j \in T} |x'_j|$$
would give us a contradiction.
\end{proof}

Note that if the $k'$ largest entries of $x'$ are picked with 
$k' > k$, then one would need to work with $T'$ for the fine
recovery stage, and therefore the starting assumptions on
$\Phi$ have to be modified for $k'$. For simplicity we shall
stick to $k'=k$ and consequently $\gamma = \sqrt{2}$.

\subsection*{Fine recovery}

Once $T$ is found, the $r$th order
Sobolev dual frame $F:=F_{\mathrm{Sob},r}$ of $E = \Phi_T$ is computed
and we set $\hat x_{\Sigma\Delta} = Fq$. We now
restate and prove Theorem~\ref{main_thm_2}.
\begin{thm}\label{main_thm_2_1} 
  Let $\Phi$ be an $m \times N$ matrix whose entries are
  i.i.d.~according to $\G(0,1)$. Suppose $\alpha \in (0,1)$ and
  $\lambda:=m/k \geq c (\log N)^{1/(1-\alpha)}$ where $c=c(r,\alpha)$. Then
  there are two constants $c'$ and $C$ that depend only on $r$ such that 
  with probability at least $1 - \exp(-c' m \lambda^{-\alpha})$ on the 
  draw of $\Phi$, the
  following holds: For every $x\in \Sigma^N_k$
  such that $\min_{j\in \mathrm{supp}(x)} |x_j| \ge C \delta$,
  the reconstruction $\hat{x}_{\sd}$ satisfies
\begin{equation}\label{sd_error_main_1}
\|x - \hat{x}_{\sd}\|_2 \lesssim_r  \lambda^{-\alpha(r-\frac{1}{2})}
\delta.
\end{equation}
\end{thm}

\begin{proof} Suppose that $\lambda \geq c (\log
  N)^{1/(1-\alpha)}$ with $c=c_6$ as in the proof of
  Theorem~\ref{sing_val_for_CS}. Let $q_\sd$ be obtained by quantizing
  $y:=\Phi x$ via an $r$th order $\sd$ scheme with alphabet $\A=\delta
  \Z$ and with the quantization rule as in \eqref{greedy-quant}, and
  let $u$ be the associated state sequence as in
  \eqref{r-th-SD-eq}. Define $x^\#$ as the solution of the program
$$\min \|z\|_1\ \text{subject to}\ \|\Phi z - q_\sd \|_2 \le
\epsilon.
$$
Suppose that $\Phi$ admits the robust recovery constant $C_1$, i.e.,
the solution $x^\#$ of the program \eqref{ell1_eps_prog_alt} satisfies
$\|x-x^\#\|_2 \leq C_1 \epsilon/\sqrt{m}$ for every $x$ in
$\Sigma^N_k$ provided that $\|y - q_\sd\| \le \epsilon$.  Note that
$C_1$, as given for example in \cite{CRT}, only depends on the RIP
constants of $\Phi$ and is well-behaved if $m$ and $N$ satisfy the
hypothesis of the theorem. As discussed in Section~\ref{sec2}, in this
case we have
$\|y-q_\sd\|_2 \le 2^{r-1} \delta \sqrt{m}$
which implies 
$$\|x-x^\#\|_2 \leq
C_1 2^{r-1} \delta.$$ 
Assume that  
\begin{equation}\label{size_cond_x_0}
\min_{j \in T} |x_j| \geq  C_1\cdot 2^{r-1/2}\delta =: C\delta.
\end{equation}
Then, Proposition~\ref{supp_recov} (with $\gamma=\sqrt{2}$ and
$\eta=C_1 2^{r-1}$) shows that $T'$, the support of the $k$
largest entries of $x^\#$, is identical to the support $T$ of $x$.
Finally, set 
$$\hat{x}_\sd=F_{\text{sob},r} q_\sd$$
where $F_{\text{sob},r}$ is the $r$th order Sobolev dual of $\Phi_T$.
Using the fact that $\|u\|_2 \le 2^{-1} \delta \sqrt{m}$ (see
Section~\ref{sec2}) together with the conclusion of
Theorem~\ref{sing_val_for_CS} and the error bound
\eqref{err-bound-Sob}, we conclude that
\begin{equation}
\|x - \hat x_{\Sigma\Delta} \|_2 \leq 
 \frac{\|u\|_2}{\sqrt{m}\,\sigma_{\mathrm{min}}(D^{-r}E)}  \leq
\frac{\lambda^{-\alpha(r-1/2)}}{2c_1}
\delta.
\end{equation}
Note that the RIP and therefore the robust recovery will hold with probability
$1-\exp(c''m)$, and our Sobolev dual reconstruction error bound will hold 
with probability $1 - \exp(-c_7 m \lambda^{-\alpha})$. Here $c_1$ and
$c_7$ are as in the proof of Theorem~\ref{sing_val_for_CS}. 

\end{proof}

\begin{rem}
  To interpret the size condition in a concrete case, assume that
  $\Phi$ admits the robust recovery constant $C_1=10$, and that we
  have
\begin{equation}\label{size_cond_x}
\min_{j \in T} |x_j| \geq \sqrt{2} \eta = 5\cdot 2^{r+1/2}\delta.
\end{equation}
If PCM is used as the quantization method, then
the best error guarantee we have that holds uniformly on $T$ would be
$$\|x - x^\#_{\mathrm{PCM}} \|_\infty \leq \|x - x^\#_{\mathrm{PCM}} \|_2 
\leq 5 \delta.$$
It can be argued that the approximately
recovered entries of $x^\#_{\mathrm{PCM}}$ are meaningful
only when the minimum nonzero entry of $x$ is at least as large
as the maximum uncertainty in $x^\#_{\mathrm{PCM}}$, which is
only known to be bounded by $5\delta$. Hence, 
in some sense the size condition \eqref{size_cond_x} is natural
(modulo the factor $2^{r+1/2}$).
\end{rem}

\subsection*{Quantizer choice and rate-distortion issues}

So far we have not made any assumptions on the step size $\delta$ of
the uniform infinite quantizer $\A = \delta \Z$. An important question concerns
how large $\delta$ should be for the most effective use of resources.
This question is motivated by the fact that infinite quantizers are not
practical and have to be replaced by finite ones. In the same vein, an alternative
question is to determine the minimum number of bits that the quantizer
needs to incorporate as well as 
the resulting approximation error. First, let us assume that 
\begin{equation}\label{dyadic_size_assumption}
 A\leq |x_j| \leq \rho := 2^b A ~~~\mbox{ for all }
j \in T.
\end{equation}
For usefulness of our results, one would be interested in
the regime $A \ll \rho$.  Thus, we introduce $2^b = \rho/A $ to
represent the number of dyadic scales over which the input is allowed
to range.  Clearly, $\delta_r$, the quantization step size used by an
$r$th order $\Sigma\Delta$ scheme for our support recovery results to
hold must satisfy $\delta_r \leq \frac{A/5}{2^{r+1/2}}$ (as before, we
assume $C_1=10$). Let us for the moment use the largest allowable
step-size, i.e., set 
\begin{equation}\label{delta_r}
\delta_r := \frac{A/5}{2^{r+1/2}}.
\end{equation}

Next, let us assume that a $B_r$-bit uniform quantizer of step size
$\delta_r$ is to replace $\A = \delta \Z$. We know that $\|q\|_\infty$
could be as large as $2^{r-1} \delta_r + \|y\|_\infty$, therefore we
need to bound $\|y\|_\infty$ efficiently.  If we use the RIP, then
$\Phi$ does not expand the $\ell_2$-norm of $k$-sparse vectors by more
than a factor of $2\sqrt{m}$ (note our choice of normalization for
$\Phi$), and therefore it follows that
$$\|y\|_\infty \leq \|y\|_2 \leq 2\sqrt{m} \|x\|_2 \leq 2  \rho\sqrt{mk},$$
which is a restatement of the inequality
$$ \|E \|_{\ell_\infty^k \to \ell_\infty^m} \leq
\sqrt{k} \|E \|_{\ell_2^k \to \ell_2^m}$$ 
that holds for any $m \times k$ matrix $E$. However, it can be argued
that the $(\infty,\infty)$-norm of a random matrix should typically be
smaller. In fact, if $E$ were drawn from the Bernoulli model, i.e.,
$E_{ij} \sim \pm 1$, then we would have
$$
\|E \|_{\ell_\infty^k \to \ell_\infty^m} = 
k = \lambda^{-1/2} \sqrt{mk},
$$
as can easily be seen from the general formula
\begin{equation}
\|E \|_{\ell_\infty^k \to \ell_\infty^m} =
 \max_{1 \leq i \leq m} \sum_{j=1}^k |E_{ij}|.
\end{equation}

Using simple concentration inequalities for Gaussian random
variables, it turns out that for the range of aspect ratio $\lambda =
m/k$ and probability of encountering a matrix $\Phi$ that we are
interested in, we have $\|E \|_{\ell_\infty^k \to \ell_\infty^m} \leq
\lambda^{-\alpha/2} \sqrt{mk}$ for every $m \times k$ submatrix $E$ of
$\Phi$.  We start with the following estimate:
\begin{prop}\label{l1_Gauss}
Let $\xi_1,\dots,\xi_k$ i.i.d.~standard Gaussian variables. Then, 
for any $\Theta  > 1$, 
\begin{equation}
\mathbb{P} \left( \sum_{j=1}^k |\xi_j| > \Theta  \right)
\leq 2^k e^{-\Theta^2 / (2k)}.
\end{equation}
\end{prop}

\begin{proof} 
\begin{eqnarray}
\mathbb{P} \left( \sum_{j=1}^k |\xi_j| > \Theta  \right)
& \leq & 
\min_{t \geq 0}~ \int_{\mathbb{R}^k} e^{-\left(\Theta  
-\sum_{j=1}^k |x_j|\right)t} \prod_{j=1}^k 
e^{-x_j^2/2}\,\frac{\mathrm{d}x_j}{\sqrt{2\pi}} \cr
& = &
\min_{t \geq 0}~ e^{- \Theta t}
\left (e^{t^2/2} \int_\R e^{-\frac{1}{2} (|x| - t)^2} 
\,\frac{\mathrm{d}x}{\sqrt{2\pi}} \right)^k\cr
& = & 
\min_{t \geq 0}~ e^{- \Theta t}
\left (2 e^{t^2/2} \int_0^\infty e^{-\frac{1}{2} (x - t)^2} 
\,\frac{\mathrm{d}x}{\sqrt{2\pi}} \right)^k\cr
& \leq &
2^k \min_{t \geq 0}~ e^{- \Theta t + k t^2/2} \cr
& = & 
2^k e^{-\Theta^2 / (2k)}.
\end{eqnarray}
where in the last step we set $t = \Theta/k$.
\end{proof}

\begin{prop}\label{inf_inf_Gauss}
  Let $\Phi$ be an $m \times N$ random matrix whose entries are
  i.i.d.~$\G(0,1)$.  Let $0 < \alpha < 1$ be any number and assume
  that
\begin{equation}\label{inf_inf_cond}
\lambda := \frac{m}{k} \geq c_1 (\log N)^{1/(1-\alpha)},
\end{equation}
where $c_1$ is an appropriate constant. Then with probability at least
$1 - e^{-c_2 m^{1-\alpha}k^\alpha}$ for some $c_2 > 0$, every $m
\times k$ submatrix $E$ of $\Phi$ satisfies
\begin{equation}
\|E \|_{\ell_\infty^k \to \ell_\infty^m} \leq 
\lambda^{-\alpha/2} \sqrt{mk}.
\end{equation}
\end{prop}

\begin{proof}
Proposition \ref{l1_Gauss} straightforwardly implies that
\begin{equation}
\mathbb{P} \left( \{ \exists T \mbox { such that } |T|=k \mbox{ and }
\| \Phi_T \|_{\ell_\infty^k \to \ell_\infty^m} > \Theta\}
\right)
\leq {N \choose k} m  2^k e^{-\Theta^2 / (2k)}.
\end{equation}
Let $\Theta = \lambda^{-\alpha/2} \sqrt{mk}$. It remains to show that
$$ k \log N  + k \log 2 + \log m + c_2 m^{1-\alpha} k^\alpha
\leq \frac{\Theta^2}{2k}.$$ 
If $c_1$ in \eqref{inf_inf_cond} is sufficiently large and $c_2$ is
sufficiently small, then the expression on the left hand side is
bounded by $ k \lambda^{1-\alpha} /2 = \Theta^2/(2k)$.
\end{proof}

\par
Without loss of generality, we may now assume that $\Phi$ also
satisfies the conclusion of Proposition \ref{inf_inf_Gauss}.  Hence we
have an improved bound on the range of $y$ given by
\begin{equation}
\|y \|_\infty \leq \rho\lambda^{-\alpha/2} \sqrt{mk}
= \rho \lambda^{(1-\alpha)/2} k.
\end{equation}
We assume $B_r$ is chosen to satisfy
\begin{equation}
\label{bit_eq}
2^{B_r-1} \delta_r = 2^{r-1}\delta_r +  \rho \lambda^{(1-\alpha)/2} k, 
\end{equation}
so that the quantizer is not overloaded. Since $\rho/\delta_r 
\approx 2^{r + 1/2 + b}$ by \eqref{dyadic_size_assumption} and
\eqref{delta_r}, we see that the second term on the 
right hand side of \eqref{bit_eq} is significantly
larger than the first, which implies
\begin{equation}\label{dominant_term}
 2^{B_r-1} \delta_r \approx 2^b A \lambda^{(1-\alpha)/2}  k.
\end{equation}
 Hence, using \eqref{delta_r} again, $B_r$ must satisfy
\begin{equation}\label{eq:rate}
2^{B_r-1}\approx 5~2^{b+r+1/2}\lambda^{(1-\alpha)/2}k.
\end{equation}
Based on Theorem \ref{main_thm_2_1}, the approximation error (the distortion) 
$\mathscr{D}_{\Sigma\Delta}$ incurred
after the fine recovery stage via Sobolev duals satisfies the bound
\begin{equation}
\mathscr{D}_{\Sigma\Delta} \lesssim_r
\lambda^{-\alpha(r-1/2)} \delta_r \approx
 \frac{ \lambda^{-\alpha( r - 1/2)}
   A}{2^{r+1/2}}. \label{eq:distortion}
\end{equation}

A similar calculation for the PCM encoder with the same step size $\delta_r$
and the standard $\ell_1$ decoder results in the necessity for
roughly the same number of bits $B_r$ as the $\Sigma\Delta$ encoder
(because of the approximation \eqref{dominant_term}), but provides only
the distortion bound
\begin{equation}
\mathscr{D}_{\mathrm{PCM}} \lesssim
 \delta_r \approx \frac{A}{2^{r+1/2}}.
 \label{eq:distortion_PCM}
 \end{equation}

Note that the analysis above requires that both PCM and $\Sigma\Delta$
encoders utilize high-resolution quantizers, however the benefit of
using $\Sigma\Delta$ encoders is obvious upon comparing
\eqref{eq:distortion} and \eqref{eq:distortion_PCM}.

\section{Numerical experiments}
\label{numerics}

In order to test the accuracy of Theorem \ref{bound_sigma_min}, our
first numerical experiment concerns the minimum singular value of
$D^{-r}E$ as a function of $\lambda = m/k$. In Figure \ref{fig1}, we
plot the worst case (the largest) value, among $1000$ realizations, of
$1/\sigma_{min}(D^{-r}E)$ for the range $1 \leq \lambda \leq 25$,
where we have kept $k=50$. As predicted by this theorem, we find that
the negative slope in the log-log scale is roughly equal to $r-1/2$,
albeit slightly less, which seems in agreement with the presence of
our control parameter $\alpha$. As for the size of the $r$-dependent constants,
the function $5^r \lambda^{-r+1/2}$
seems to be a reasonably close numerical fit, which also explains why
we observe the separation of the individual curves after $\lambda > 5$.

Our next experiment involves the full quantization algorithm for
compressed sensing including the ``recovery of support'' and ``fine
recovery'' stages. To that end, we first generate a $1000\times 2000$
matrix $\Phi$, where the entries of $\Phi$ are drawn i.i.d. according
to $\G(0,1)$. To examine the performance of the proposed scheme as the
redundancy $\lambda$ increases in comparison to the performance of the
standard PCM quantization, we run a set of experiments: In each
experiment we fix the sparsity $k\in\{5,10,20,40\}$, and we
generate $k$-sparse signals $x$ with the non-zero entries of each
signal supported on a random set $T$, but with magnitude
$1/\sqrt{k}$. This ensures that $\|x\|_2=1$. Next, for
$m\in\{100,200,...,1000\}$ we generate the measurements $y=\Phi^{(m)}
x$, where $\Phi^{(m)}$ is comprised of the first $m$ rows of
$\Phi$. We then quantize $y$ using PCM, as well as the $1$st and $2$nd
order $\Sigma\Delta$ quantizers, defined via \eqref{r-th-SD-eq} and
\eqref{greedy-quant} (in all cases the quantizer step size is
$\delta=10^{-2}$). For each of these quantized measurements $q$, we
perform the coarse recovery stage, i.e., we solve the associated
$\ell_1$ minimization problem to recover a coarse estimate of $x$ as
well as an estimate $\widetilde{T}$ of the support $T$. The
approximation error obtained using the coarse estimate (with PCM
quantization) is displayed in Figures \ref{fig2} and \ref{fig3} (see
the dotted curve). Next, we implement the fine recovery stage of our
algorithm. In particular, we use the estimated support set
$\widetilde{T}$ and generate the associated dual
$F_{\text{sob},r}$. Defining
$F_{\text{sob},0}:=(\Phi_{\widetilde{T}}^{(m)})^\dagger$, in each
case, our final estimate of the signal is obtained via the fine
recovery stage as $\hat{x}_{\tilde{T}}=F_{\text{sob},r}q$,
$\hat{x}_{\tilde{T^c}}=0$. Note that this way, we obtain an
alternative reconstruction also in the case of PCM. We repeat this
experiment $100$ times for each $(k,m)$ pair and plot the average of
the resulting errors $\|x-\tilde{x}\|_2$ as a function of $\lambda$ in
Figure \ref{fig2} as well as the maximum of $\|x-\hat{x}\|_2$ in
Figure \ref{fig3}.  For our final experiment, we choose the entries of
$x_T$ i.i.d. from $\G(0,1)$, and use a quantizer step size
$\delta=10^{-4}$.  Otherwise, the experimental setup is identical to
the previous one. The average of the resulting errors
$\|x-\tilde{x}\|_2$ as a function of $\lambda$ is reported in Figure
\ref{fig4} and the maximum of $\|x-\hat{x}\|_2$ in Figure \ref{fig5}.

The main observations that we obtain from these experiments are as follows:
\begin{itemize}
\item $\sd$ schemes outperform the coarse reconstruction obtained from
  PCM quantized measurements significantly even when $r=1$ and even
  for small values of $\lambda$.
\item For the $\Sigma\Delta$ reconstruction error, the negative slope in the
  log-log scale is roughly equal to $r$. This outperforms the (best
  case) predictions of Theorem \ref{main_thm_2} which are obtained 
  through the operator
  norm bound and suggests the presence of 
  further cancellation due to the statistical
  nature of the $\Sigma\Delta$ state variable $u$, similar to the
  white noise hypothesis. 
\item When a fine recovery stage is employed in the case of PCM
  (using the Moore-Penrose pseudoinverse of the submatrix of $\Phi$
  that corresponds to the estimated support of $x$), the approximation
  is consistently improved (when compared to the coarse
  recovery). Moreover, the associated approximation error is observed to be of
  order $O(\lambda^{-1/2})$, in contrast with the error corresponding
  to the coarse recovery from PCM quantized measurements (with the 
  $\ell_1$ decoder only) where the
  approximation error does not seem to depend on $\lambda$. A
  rigorous analysis of this behaviour will be given in a separate
  manuscript.
\end{itemize}

\begin{figure}[b] 
\centerline{\includegraphics[width=5in]{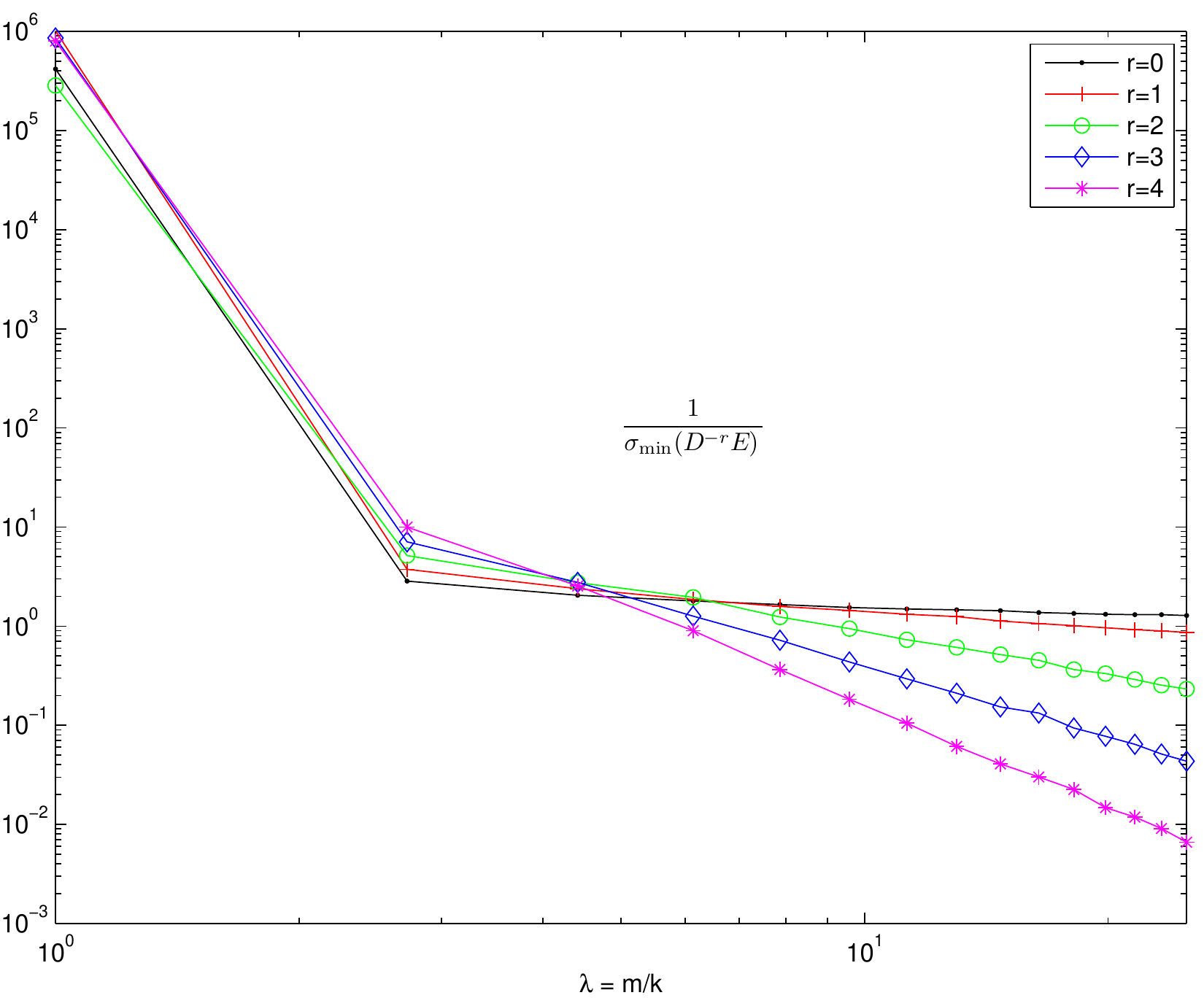}}
\caption{Numerical behavior (in log-log scale) of $1/\sigma_{min}(D^{-r}E)$
as a function of $\lambda = m/k$, for $r=0,1,2,3,4$. In this figure,
$k=50$ and $1 \leq \lambda \leq 25$. 
For each problem size, the largest value of 
$1/\sigma_{min}(D^{-r}E)$ among
$1000$ realizations of a random $m \times k$
matrix $E$ sampled from the Gaussian ensemble $\G(0,\frac{1}{m}I_m)$
was recorded.}
\label{fig1}
\end{figure}

\begin{figure}[b]
\begin{center}
\subfigure[]{
\includegraphics[width=3in,height=2in]{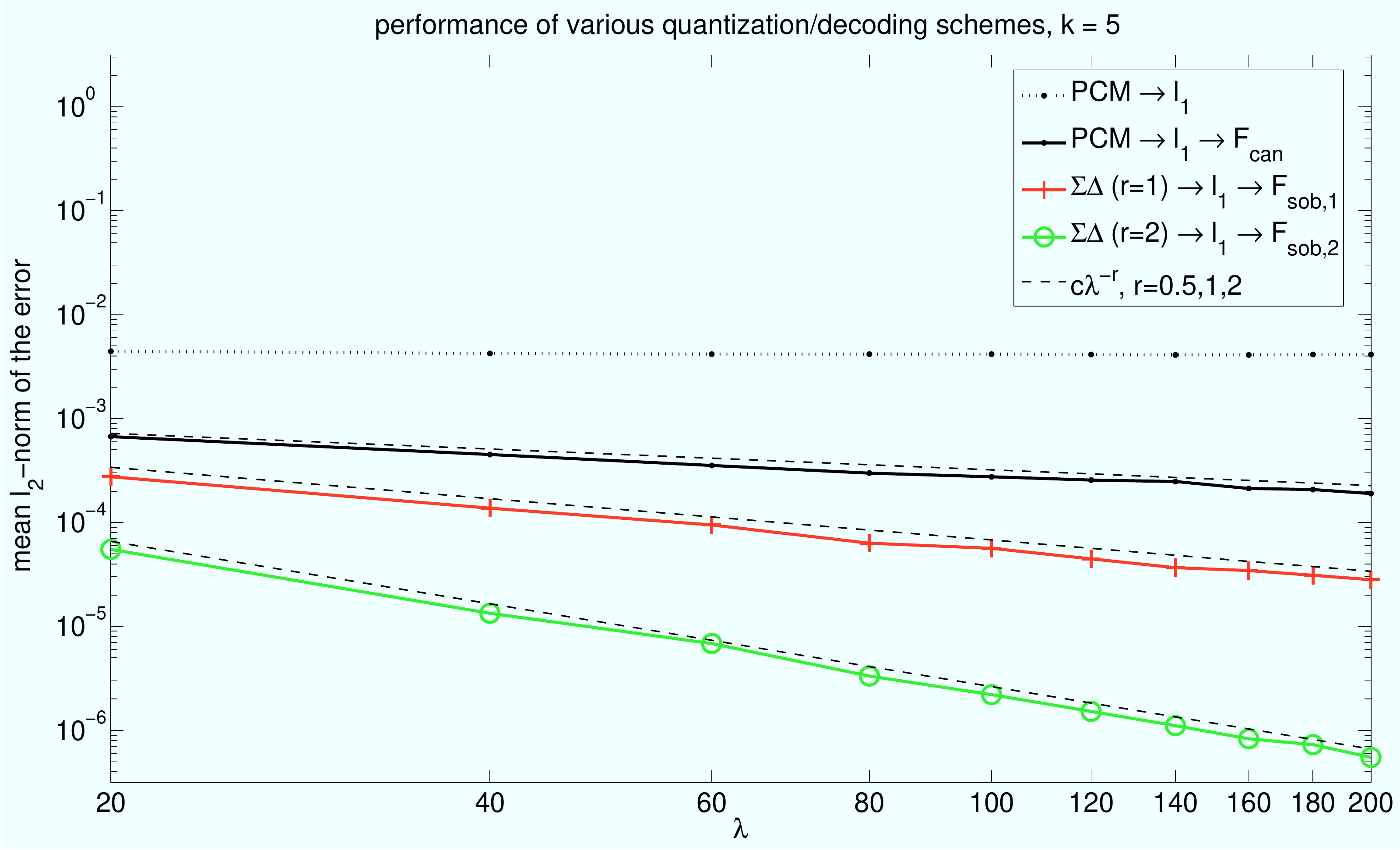}}
\subfigure[]{
\includegraphics[width=3in,height=2in]{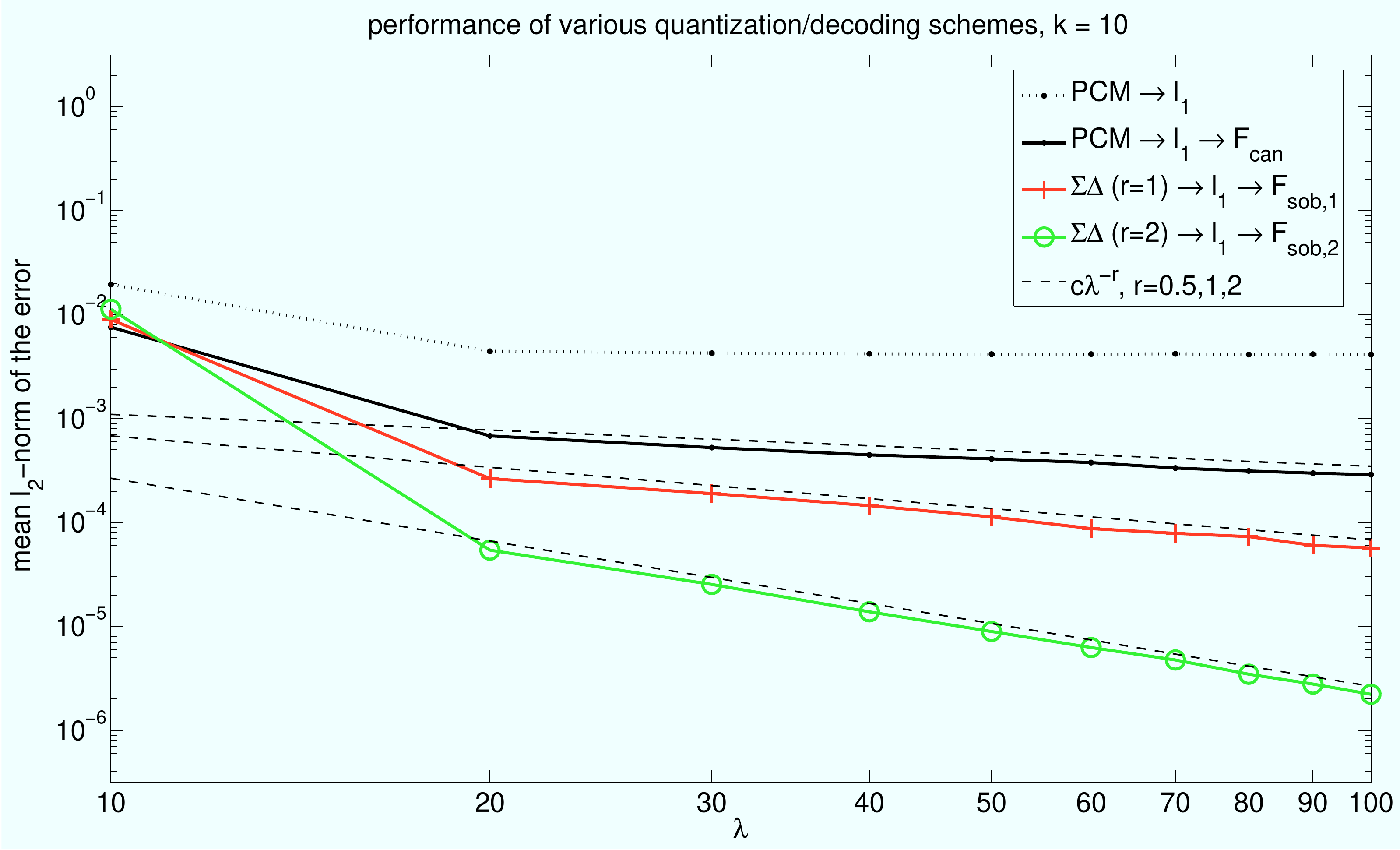}}
\subfigure[]{
\includegraphics[width=3in,height=2in]{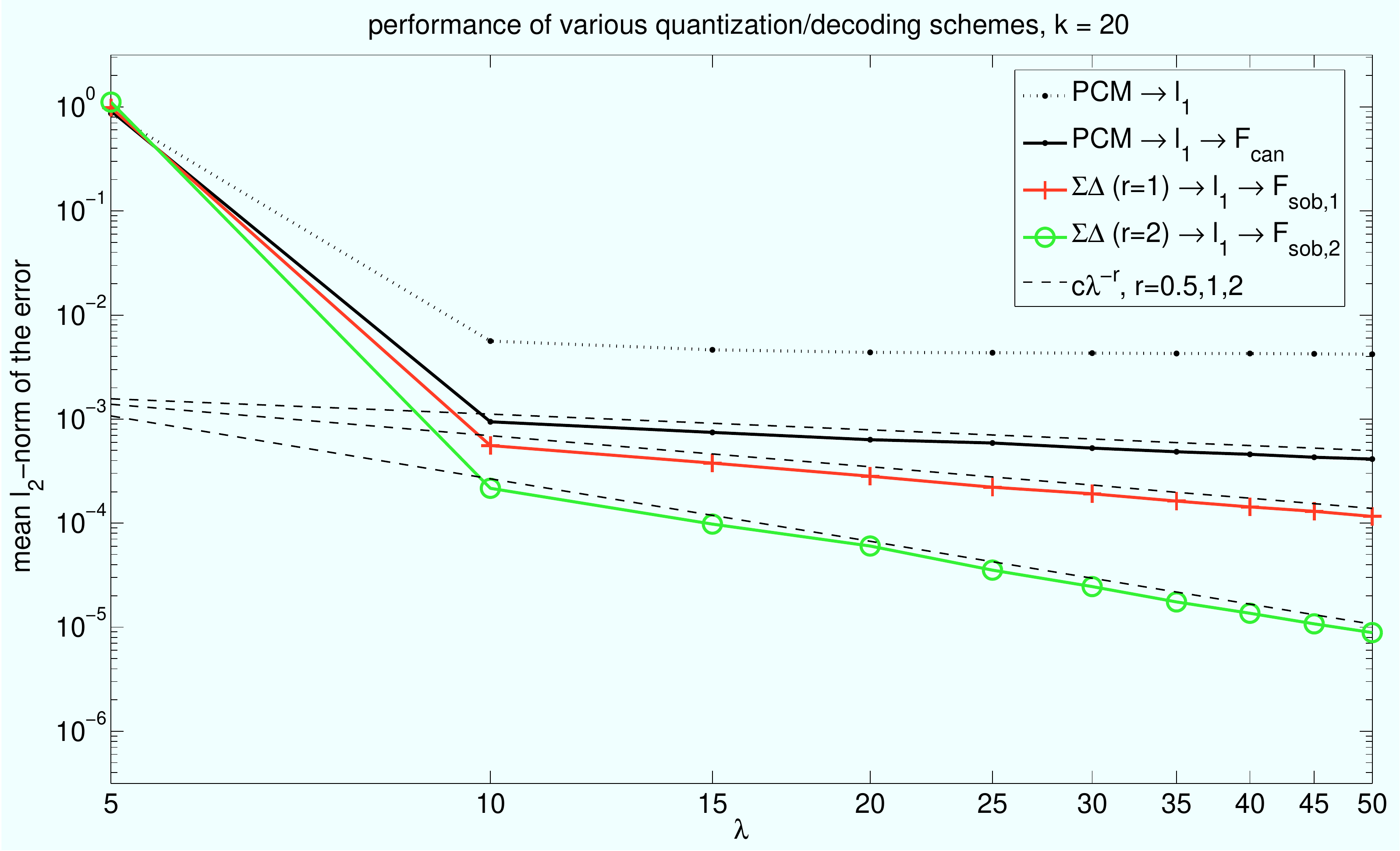}}
\subfigure[]{
\includegraphics[width=3in,height=2in]{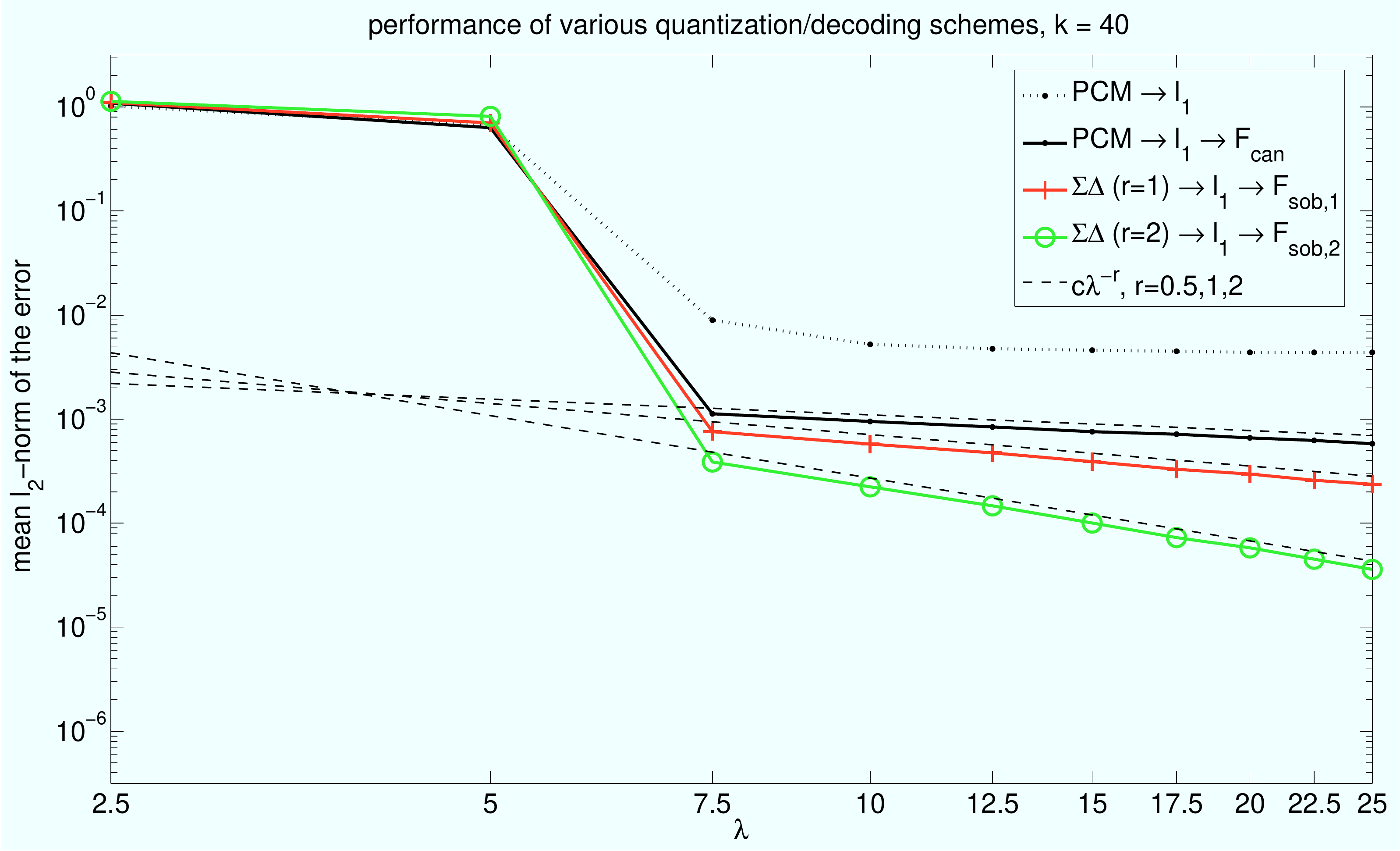}}
\caption{\label{fig2} The average performance of the proposed $\Sigma\Delta$ quantization and reconstruction schemes for various values of $k$. 
For this experiment the non-zero entries of $x$ are constant and $\delta=0.01$.}
\end{center}
\end{figure}

\begin{figure}[b]
\begin{center}
\subfigure[]{
\includegraphics[width=3in,height=2in]{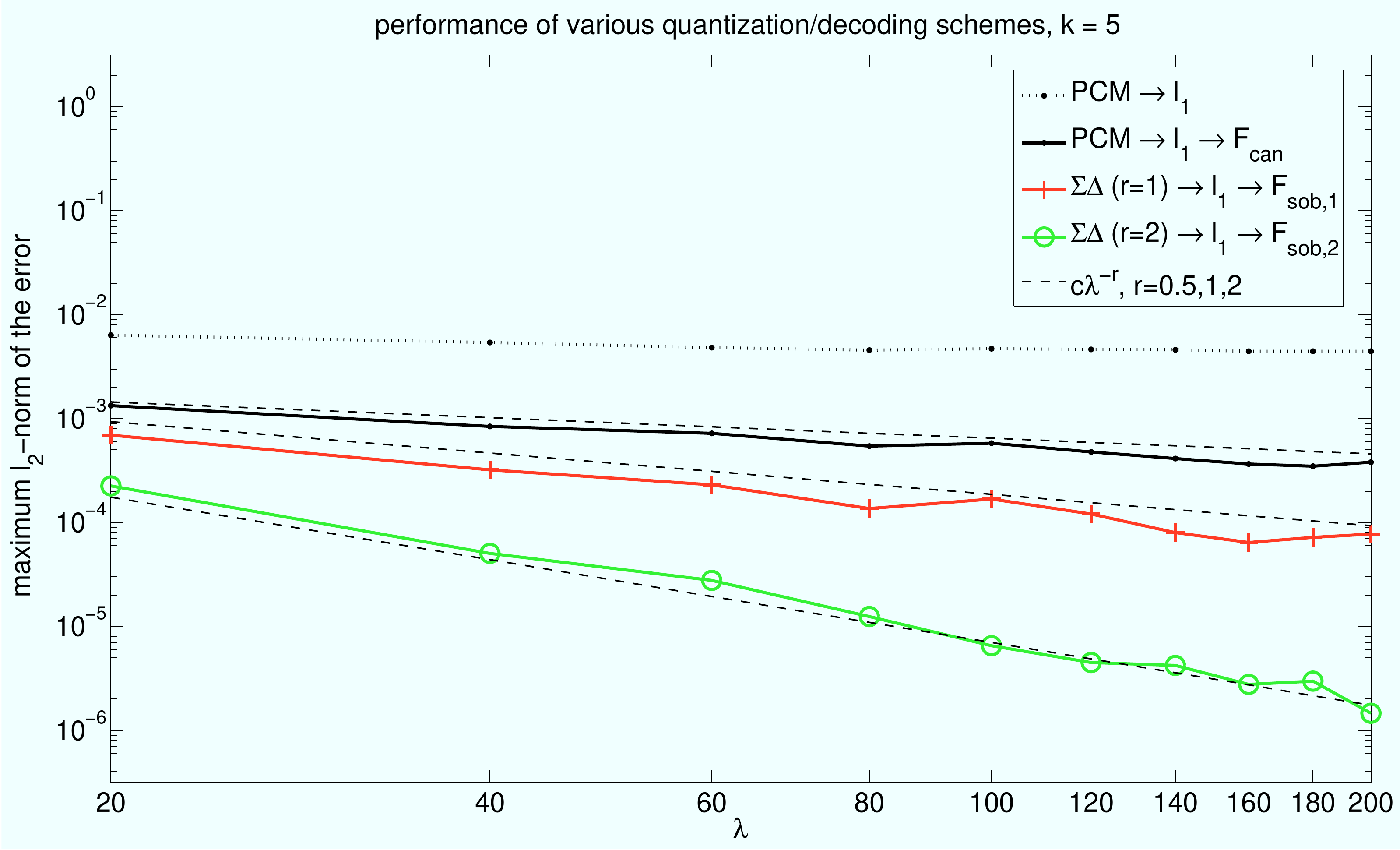}}
\subfigure[]{
\includegraphics[width=3in,height=2in]{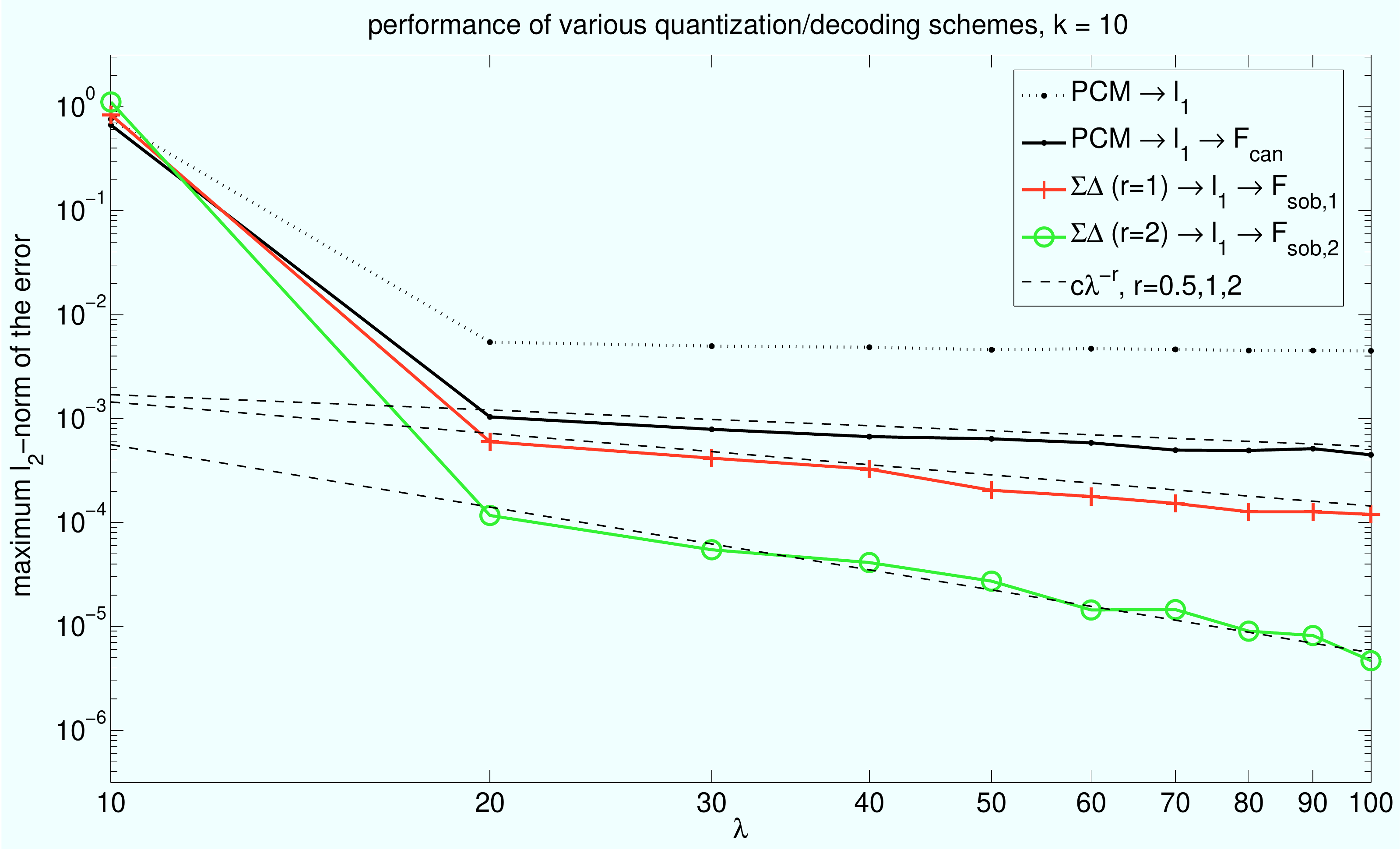}}
\subfigure[]{
\includegraphics[width=3in,height=2in]{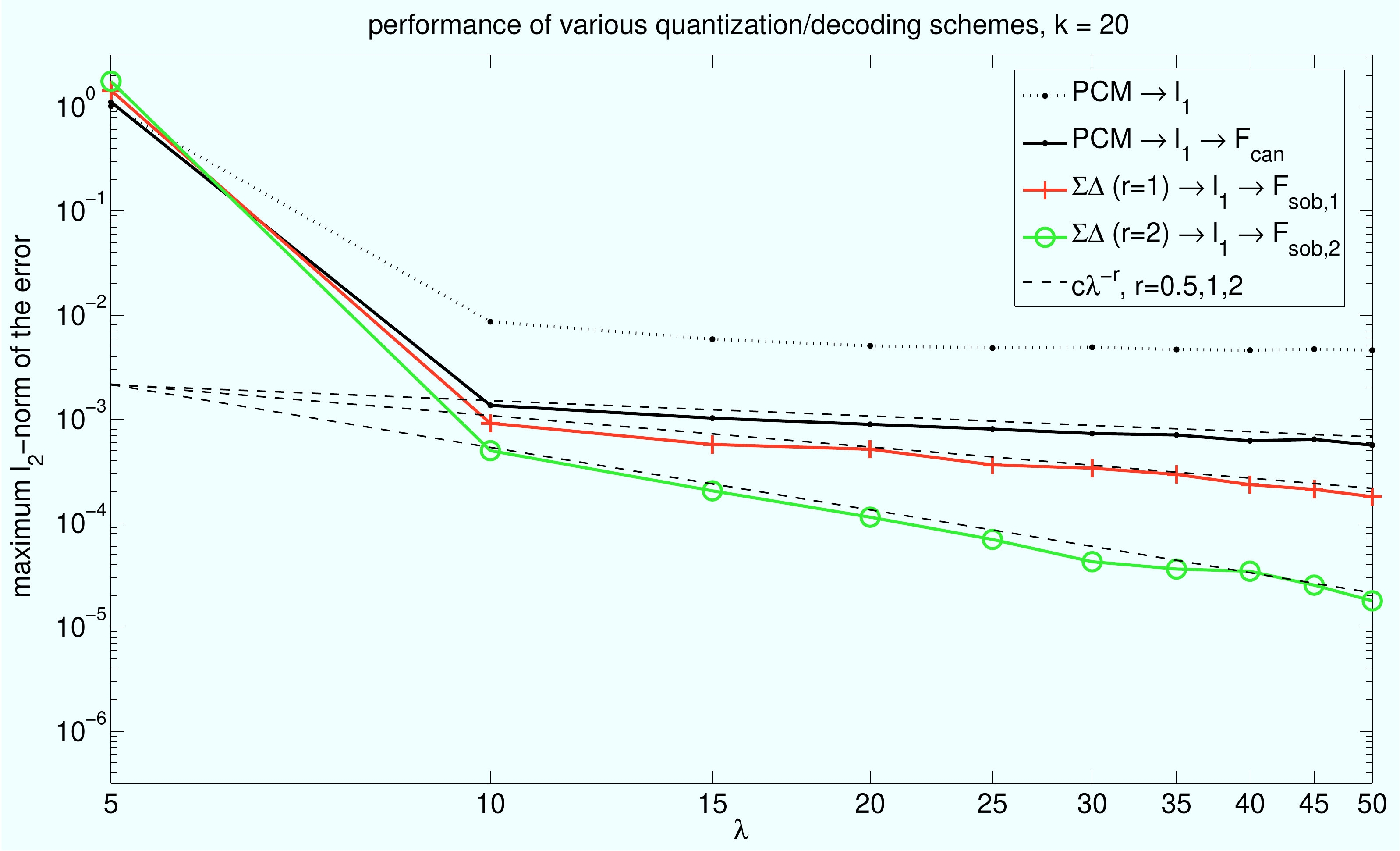}}
\subfigure[]{
\includegraphics[width=3in,height=2in]{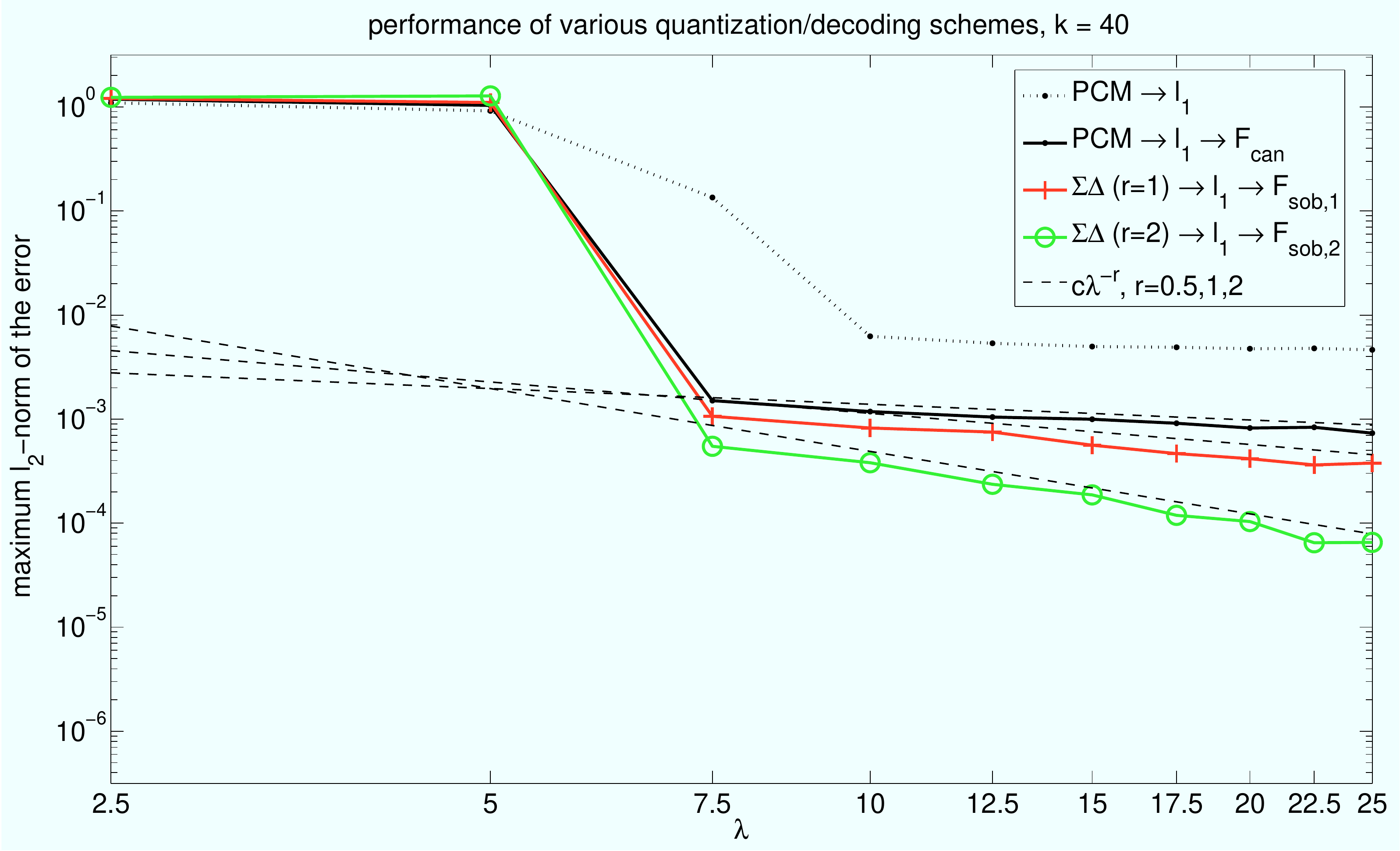}}
\caption{\label{fig3} The worst case performance of the proposed $\Sigma\Delta$ quantization and reconstruction schemes for various values of $k$. 
For this experiment the non-zero entries of $x$ are constant and $\delta=0.01$.}
\end{center}
\end{figure}

\begin{figure}[b]
\begin{center}
\subfigure[]{
\includegraphics[width=3in,height=2in]{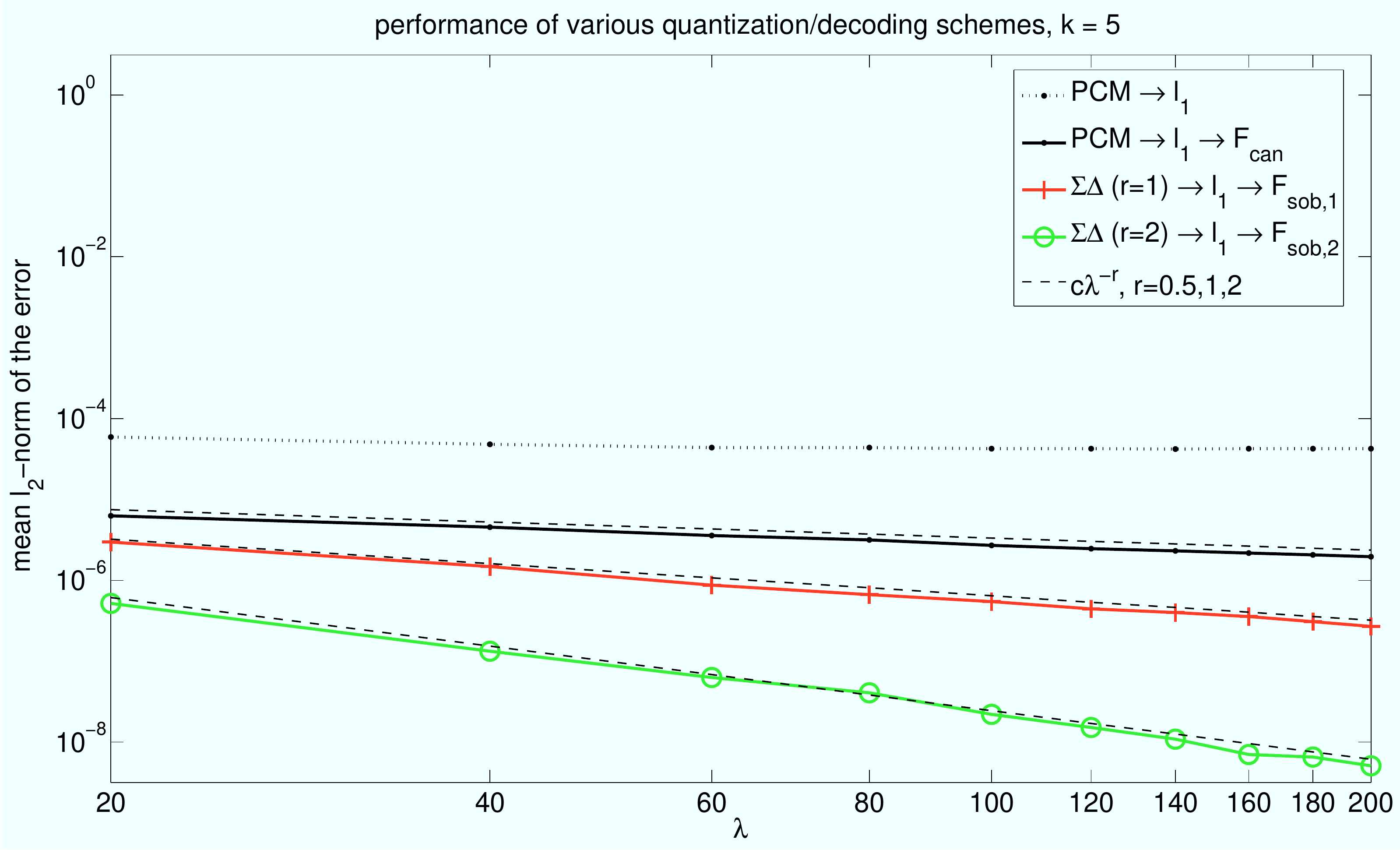}}
\subfigure[]{
\includegraphics[width=3in,height=2in]{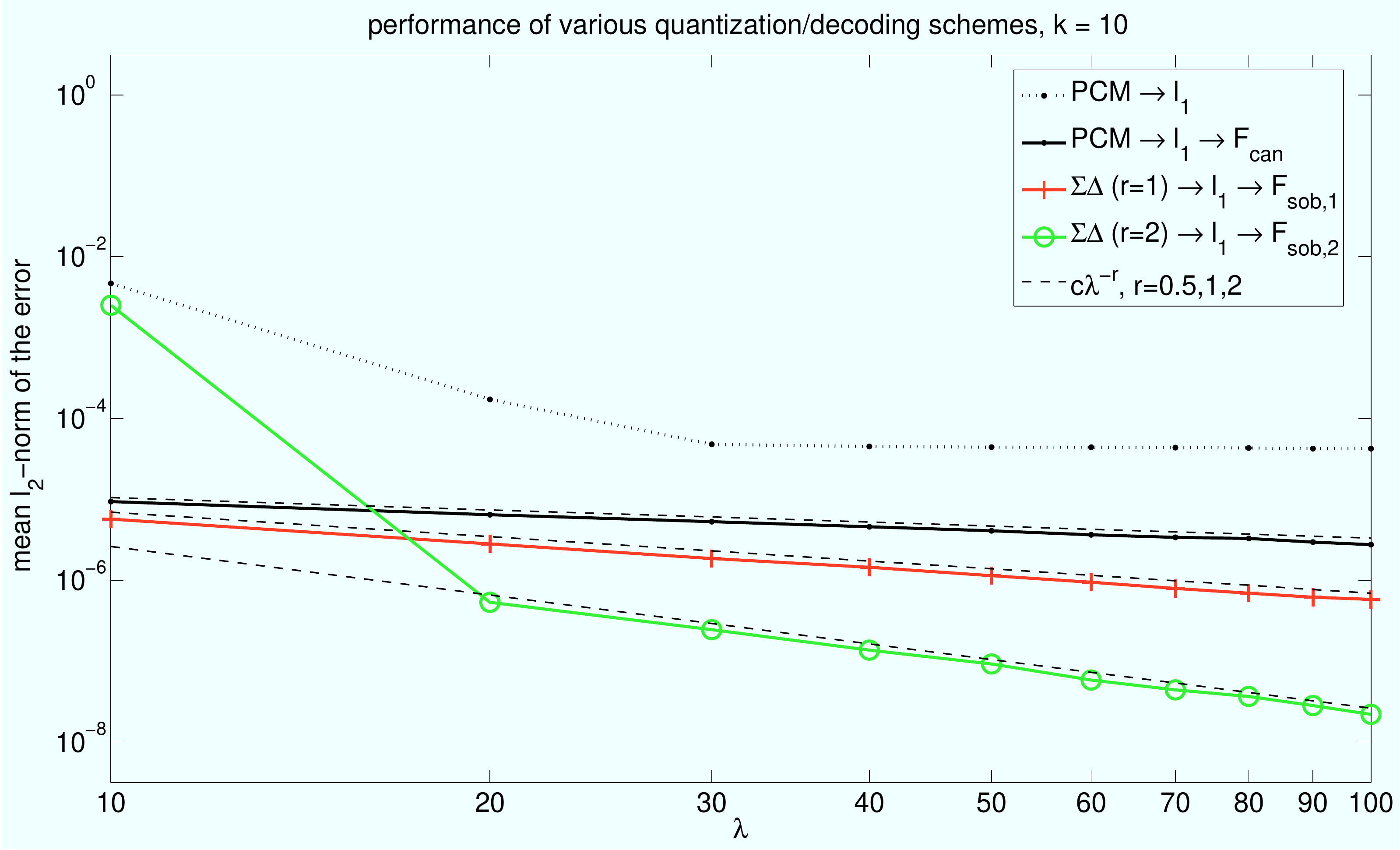}}
\subfigure[]{
\includegraphics[width=3in,height=2in]{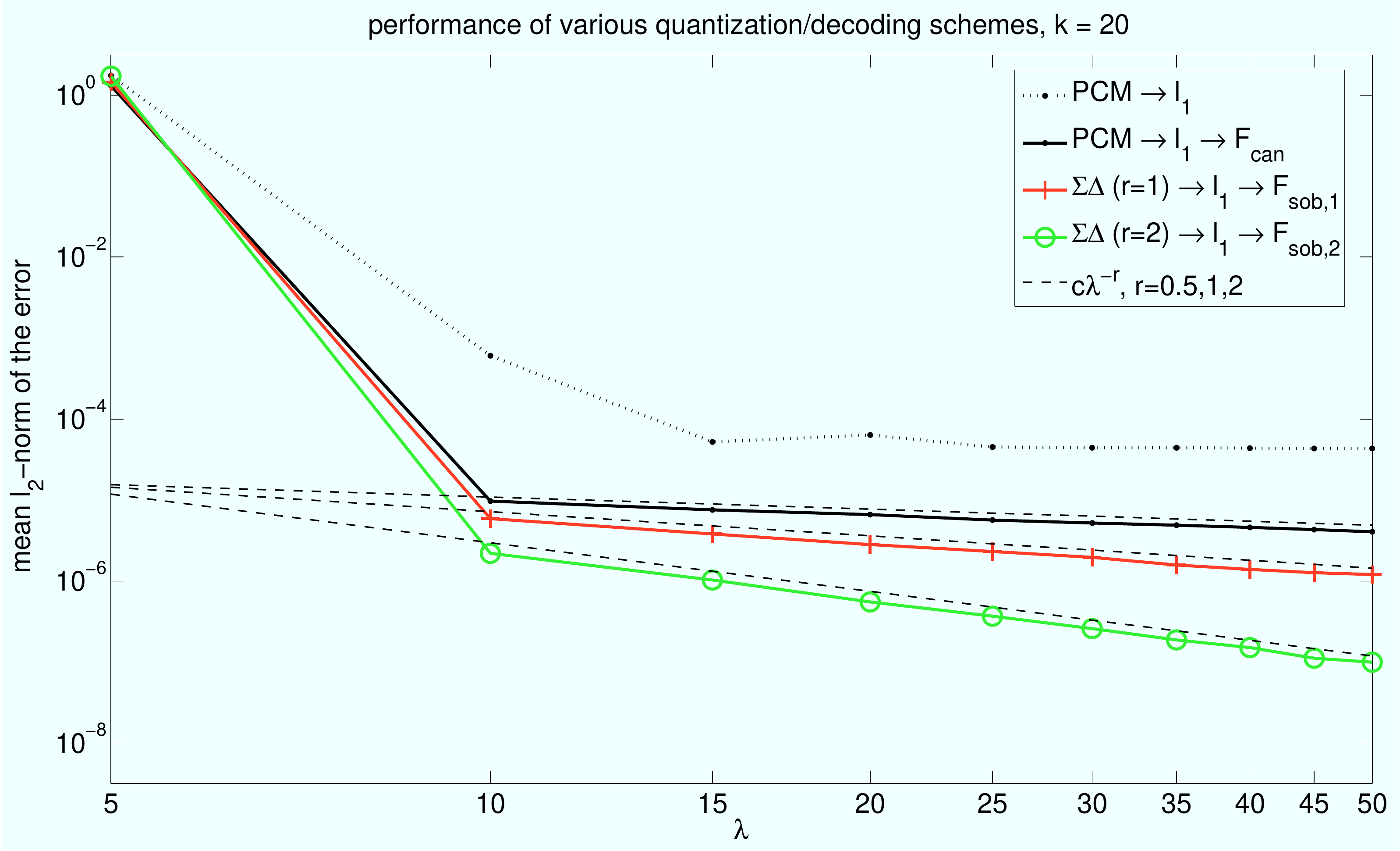}}
\subfigure[]{
\includegraphics[width=3in,height=2in]{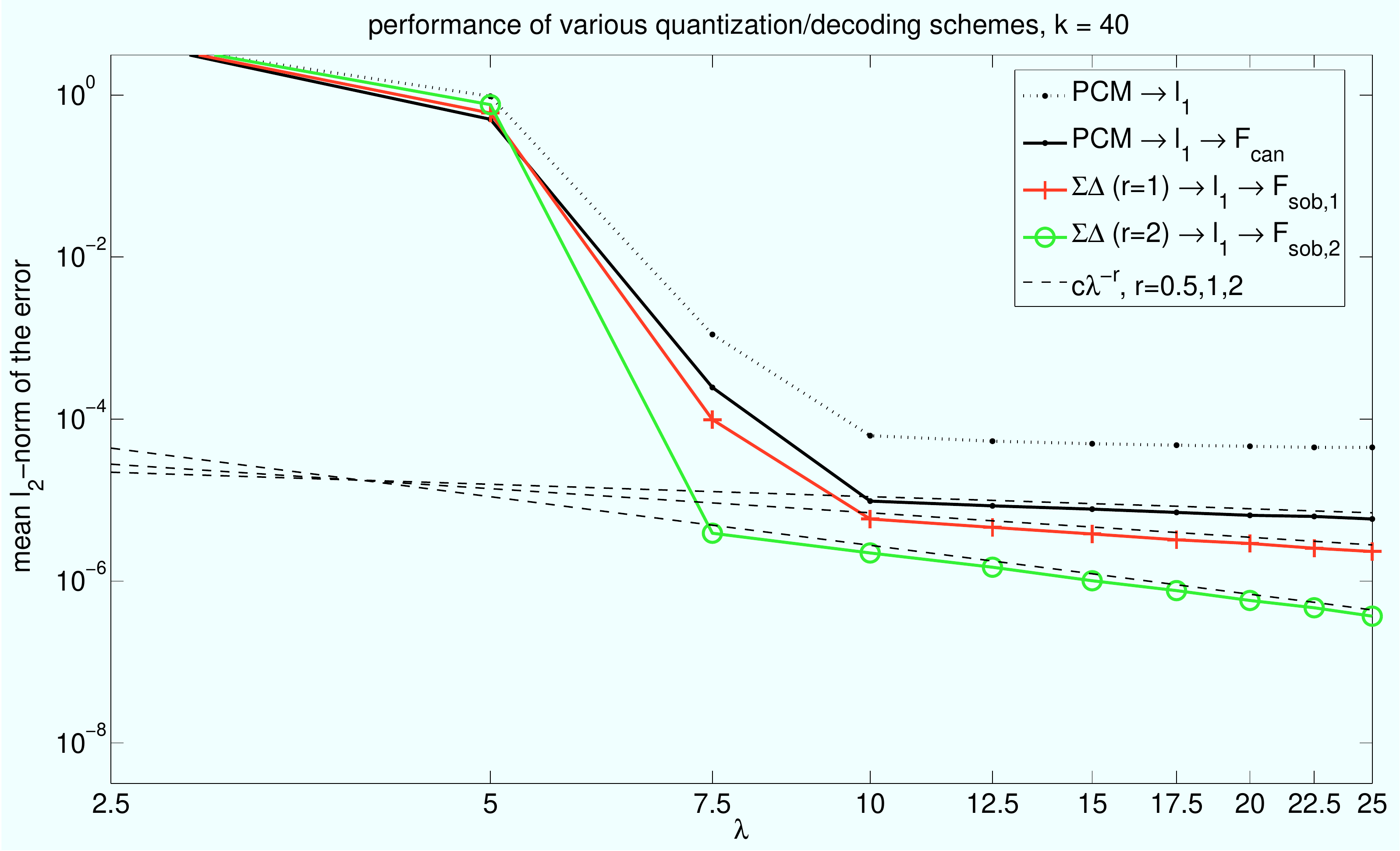}}
\caption{\label{fig4} The average performance of the proposed $\Sigma\Delta$ quantization and reconstruction schemes for various values of $k$. 
For this experiment the non-zero entries of $x$ are i.i.d. $\G(0,1)$ and $\delta=10^{-4}$.}
\end{center}
\end{figure}

\begin{figure}[b]
\begin{center}
\subfigure[]{
\includegraphics[width=3in,height=2in]{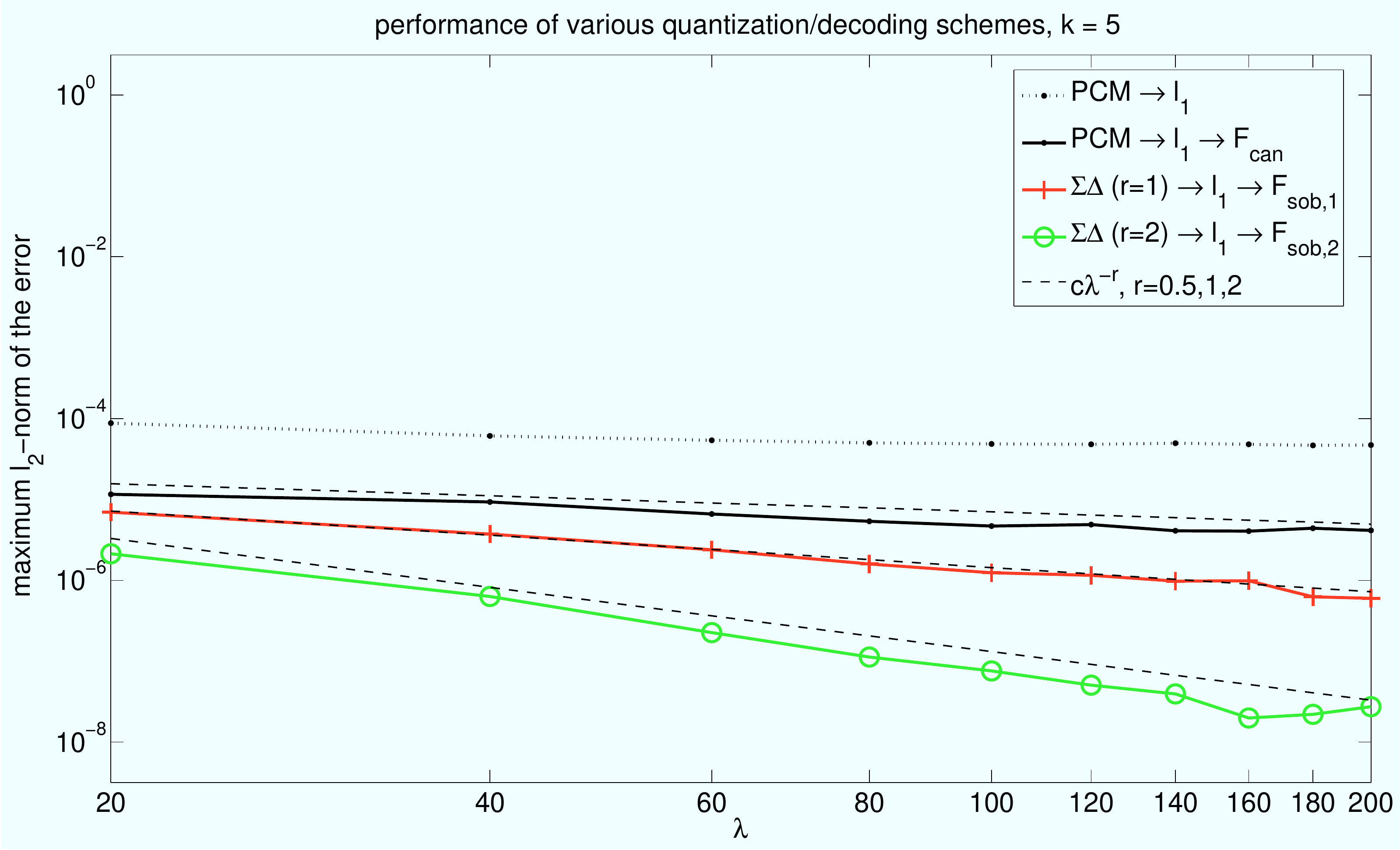}}
\subfigure[]{
\includegraphics[width=3in,height=2in]{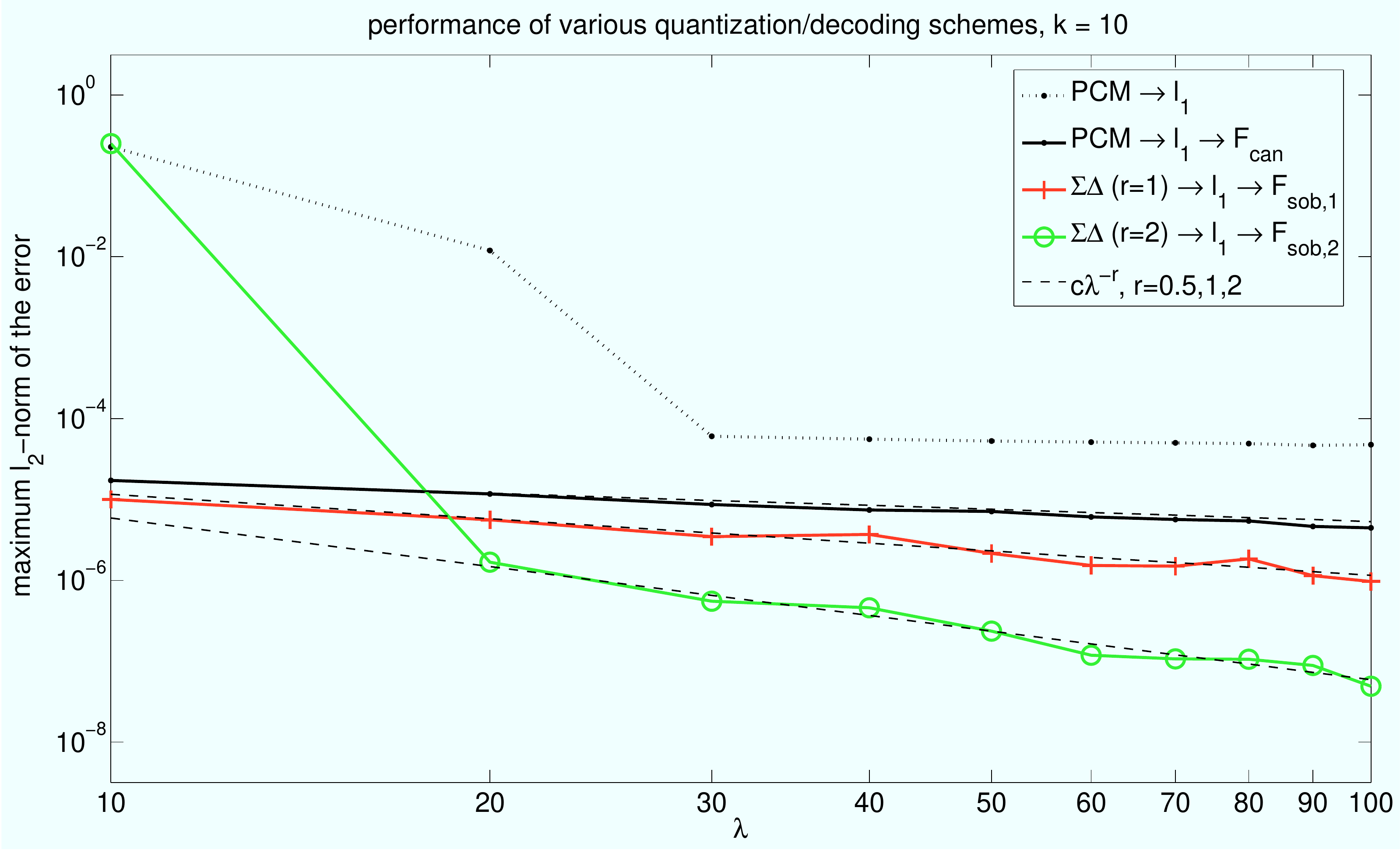}}
\subfigure[]{
\includegraphics[width=3in,height=2in]{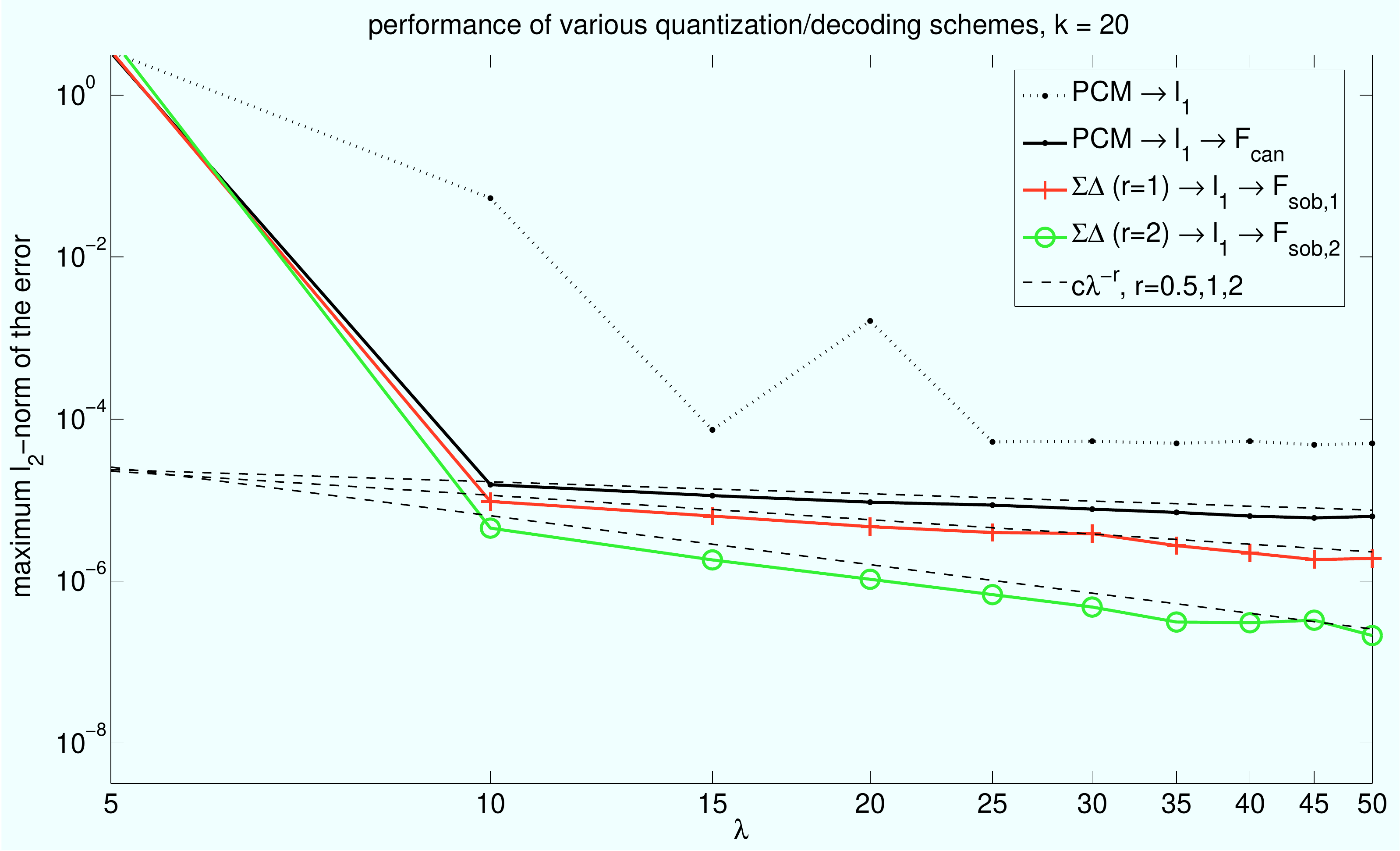}}
\subfigure[]{
\includegraphics[width=3in,height=2in]{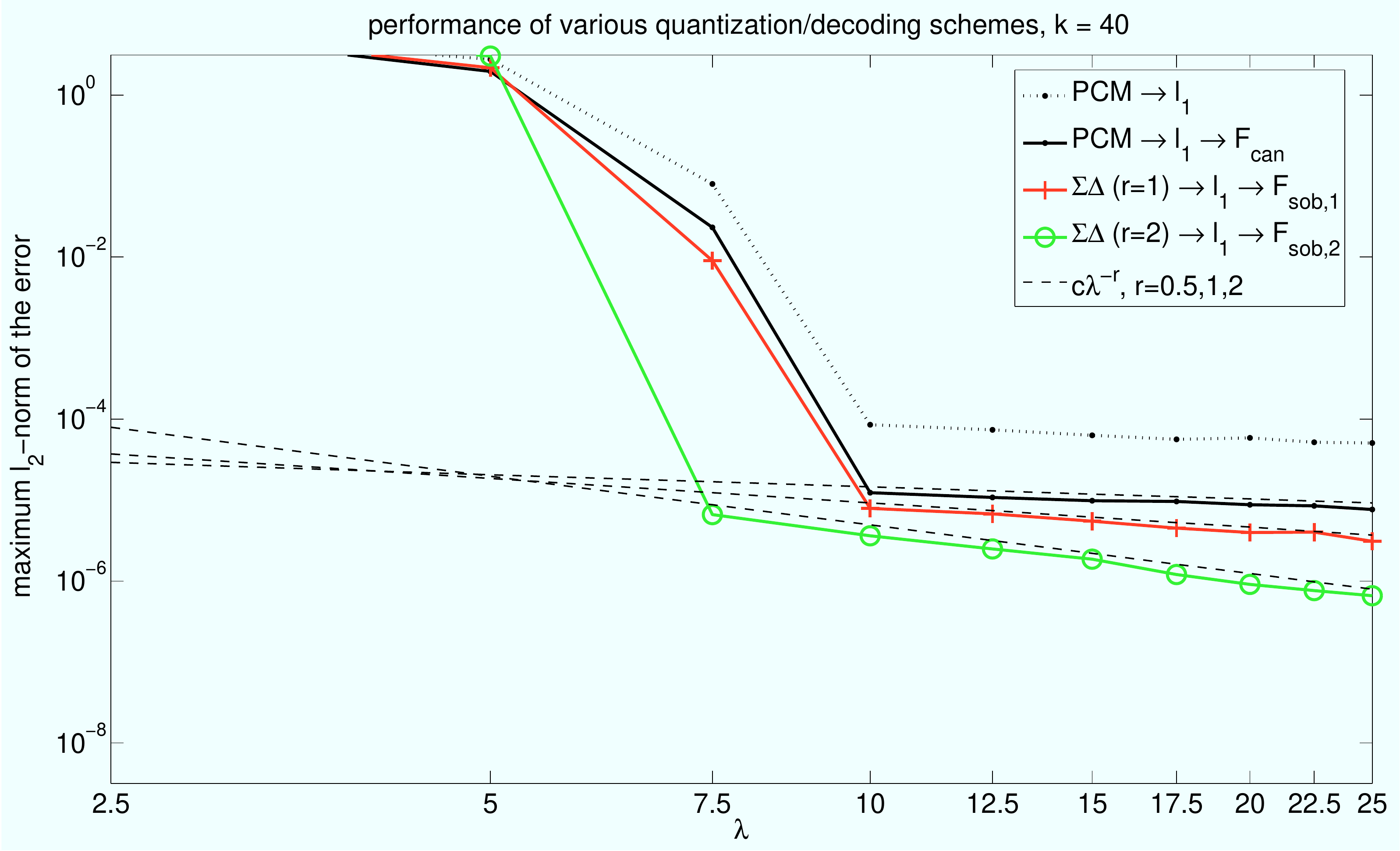}}
\caption{\label{fig5} The worst case performance of the proposed $\Sigma\Delta$ quantization and reconstruction schemes for various values of $k$. 
For this experiment the non-zero entries of $x$ are i.i.d. $\G(0,1)$ and $\delta=10^{-4}$.}
\end{center}
\end{figure}

\section{Remarks on extensions}
\label{extensions}

\subsection{Other noise shaping matrices}
In the above approach, the particular quantization scheme that we use
can be identified with its ``noise-shaping matrix'', which is $D^r$ in
the case of an $r$th order $\sd$ scheme and the identity matrix in the
case of PCM. 

The results we obtained above are valid for the aforementioned
noise-shaping matrices. However, our techniques are fairly general and
our estimates can be modified to investigate the accuracy obtained
using an arbitrary quantization scheme with the associated invertible
noise-shaping matrix $H$. In particular, the estimates depend solely
on the distribution of the singular values of $H$. Of course, in this
case, we also need change our ``fine recovery'' stage and use the
``$H$-dual'' of the corresponding frame $E$, which we define via
\begin{equation}\label{H-dual-formula}
F_{H} H = (HE)^\dagger.
\end{equation}

As an example, consider an $r$th order {\em high-pass} $\sd$ scheme
whose noise shaping matrix is $H^r$ where $H$ is defined via 
\begin{equation}\label{def-H}
H_{ij} := \left \{ 
\begin{array}{rl}
1, & \mbox{if $i=j$ or if $i=j+1$}, \cr
0, & \mbox{otherwise.}
\end{array}
\right.
\end{equation} 
It is easy to check that the singular values of $H$ are identical to
those of $D$. It follows that all the results presented in this paper
are valid also if the compressed measurements are quantized via an an
$r$th order high-pass $\sd$ scheme, provided the reconstruction is
done using the $H^r$-duals instead of the $r$th order Sobolev
duals. Note that such a result for high-pass $\sd$ schemes is not
known to hold in the case of structured frames.

\subsection{Measurement noise and compressible signals}
One of the natural questions is whether the quantization methods
developed in this paper are effective in the presence of measurement
noise in addition to the error introduced during the quantization
process. Another natural question is how to extend 
this theory to include the case when the
underlying signals are not necessarily strictly sparse, but 
nevertheless still ``compressible''. 

Suppose $x\in \R^N$ is not sparse, but compressible in the usual sense
 (e.g. as in \cite{CRT}), and let $y=\Phi x+e$, where $e$ stands for 
additive measurement noise.  The {\em coarse
recovery stage} inherits the stability and robustness properties
of $\ell_1$ decoding for 
compressed sensing, therefore the accuracy of this first reconstruction
depends on the best $k$-term approximation error for $x$,
and the deviation of $\Phi x$ from the quantized signal $q$ (which 
comprises of the measurement noise $e$ and the quantization error
$y - q$). Up to constant factors, the quantization error
for any (stable) $\sd$ quantizer is comparable to that of PCM, 
hence the reconstruction error at the coarse recovery stage would also
be comparable. In the {\em fine recovery stage}, however, the
difference between
$\sigma_{\max}(F_{H} H)$ and $\sigma_{\max}(F_{H})$ plays a critical
role. In the particular case of $H=D^r$ and $F_H=F_{\text{sob},r}$, 
the Sobolev duals we use in the
reconstruction are tailored to reduce the effect of the quantization error
introduced by an $r$th order $\sd$ quantizer. This is
reflected in the fact that as $\lambda$ increases, the kernel of
the reconstruction operator  $F_{\text{sob},r}$
contains a larger portion of high-pass sequences (like  
the quantization error of $\Sigma\Delta$ modulation), and is quantified by the 
bound $\sigma_{\max}(F_{\text{sob},r}D^r) \lesssim
\lambda^{-(r-1/2)}m^{-1/2}$ (see
Theorem~\ref{main_thm_1}, \eqref{sobolev-formula} and 
\eqref{err-bound-Sob}). Consequently, obtaining more measurements
increases $\lambda$, and even though $\|y - q\|_2$ increases as well,
the reconstruction error due to quantization decreases. 
At the same time, obtaining more measurements would also increase
the size of the external noise $e$, as well as the ``aliasing error'' 
that is the result of the ``off-support'' entries of $x$. However, 
this noise+error term is not counteracted by the
action of $F_{\text{sob},r}$. In fact, 
for any dual $F$, the relation $ F E = I$ implies
$\sigma_{\max}(F) \geq 1/\sigma_{\max}(E) \gtrsim
m^{-1/2}$ already and in the case of measurement noise, 
it is not possible to do better than the canonical dual $E^\dagger$ on average. 
In this case, depending on the size of the noise term, 
the fine recovery stage may not improve
the total reconstruction error even though the ``quantizer error'' is still
reduced.

One possible remedy for this problem is to construct alternative
quantization schemes with associated noise-shaping matrices that
balance the above discussed trade-off between the quantization error
and the error that is introduced by other factors. This is a delicate
procedure, and it will be investigated thoroughly in future
work. However, a first such construction can be made by using
``leaky'' $\sd$ schemes with $H$ given by 
\begin{equation}\label{def-lH}
H_{ij} := \left \{ 
\begin{array}{rl}
1, & \mbox{if $i=j$}, \cr
-\mu & \mbox{if $i=j+1$}, \cr
0, & \mbox{otherwise,}
\end{array}
\right.
\end{equation} 
where $\mu \in (0,1)$. Our preliminary numerical experiments 
(see Figures \ref{fig6} and \ref{fig7}) suggest
that this approach can be used to improve the accuracy of the
approximation further in the fine recovery stage in this more general
setting. We note that the parameter $\mu$ above can be adjusted based
on how compressible the signals of interest are and what the expected
noise level is.

\begin{figure}[b]
\begin{center}
\subfigure[]{
\includegraphics[width=3in,height=2in]{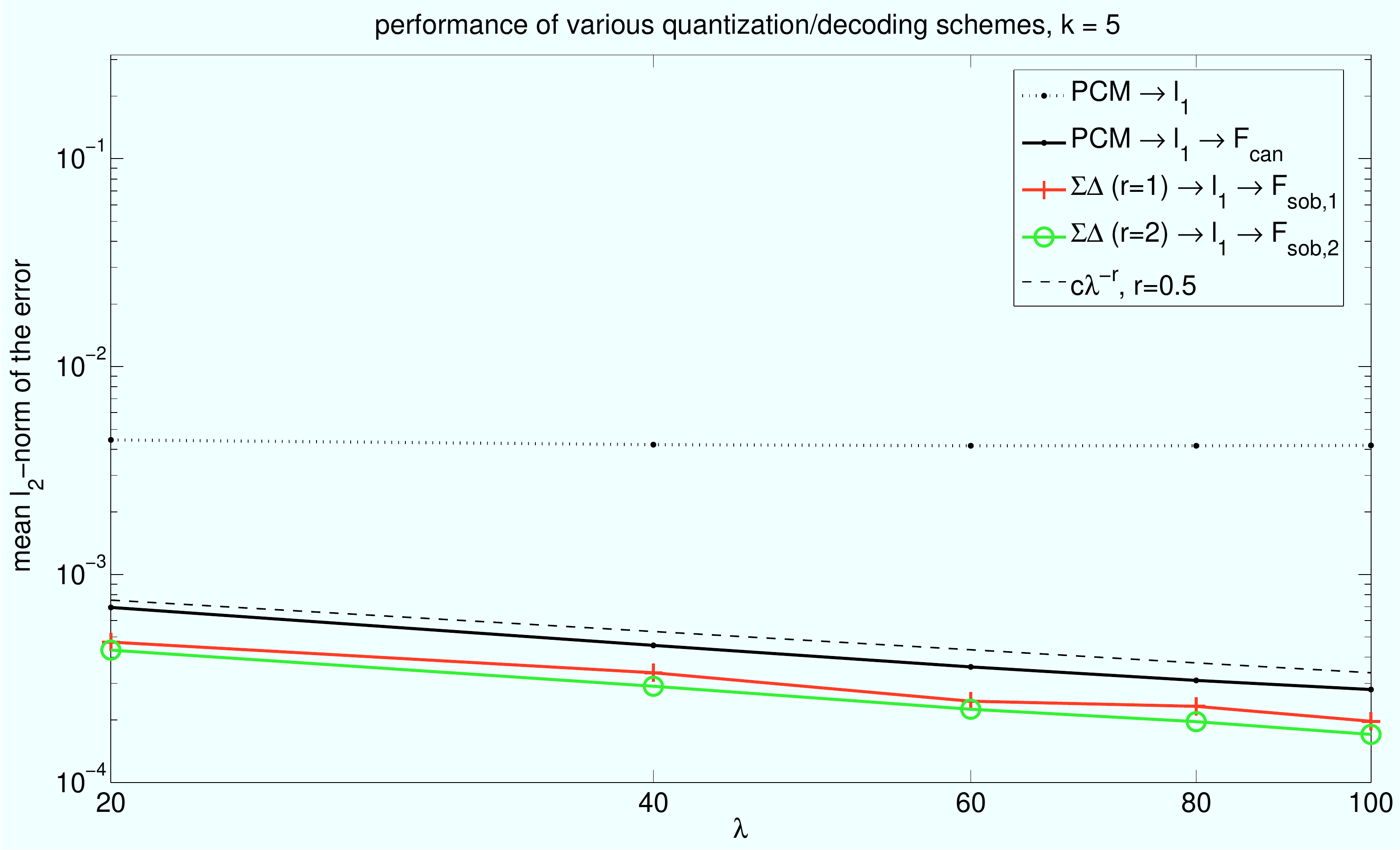}}
\subfigure[]{
\includegraphics[width=3in,height=2in]{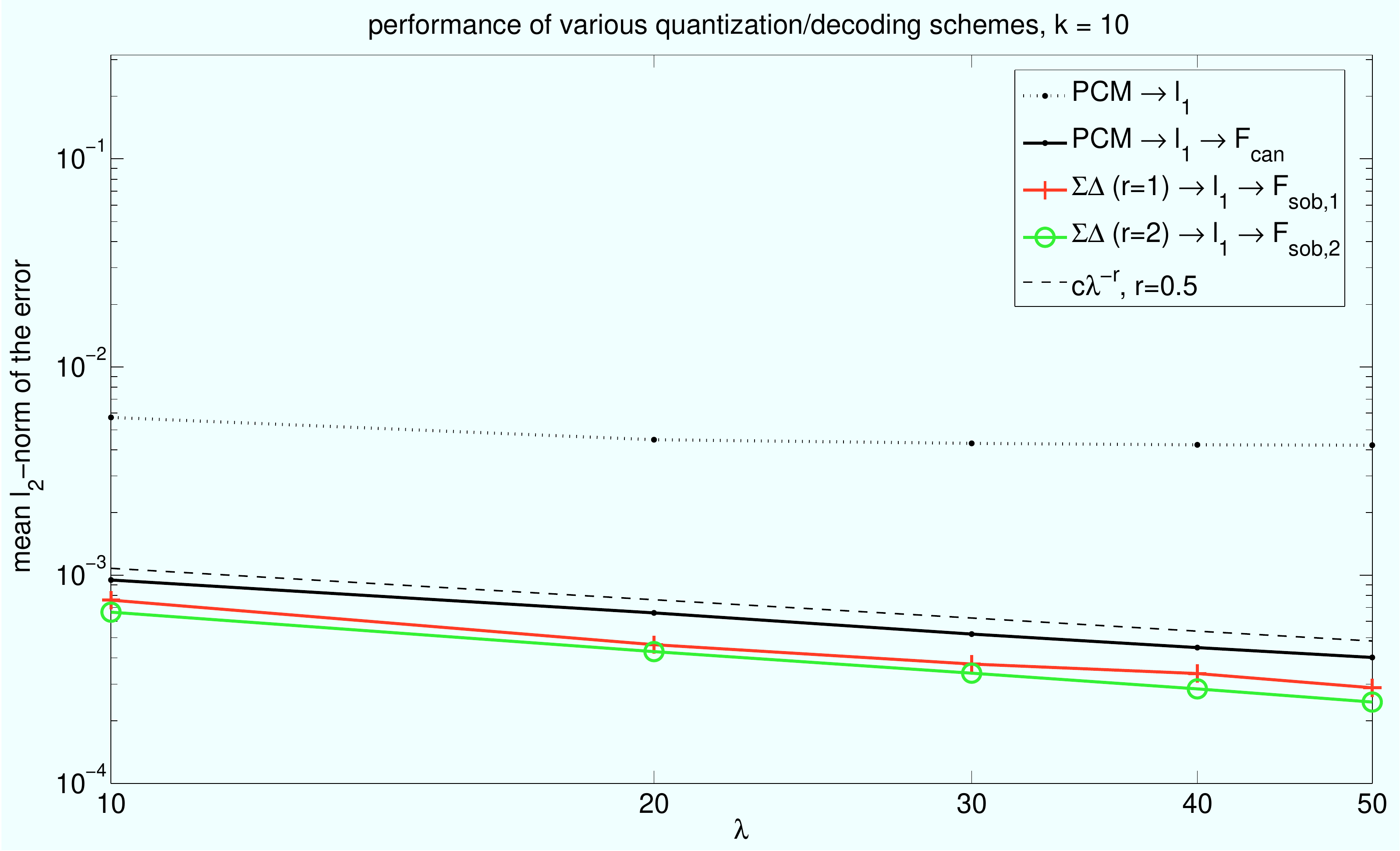}}
\subfigure[]{
\includegraphics[width=3in,height=2in]{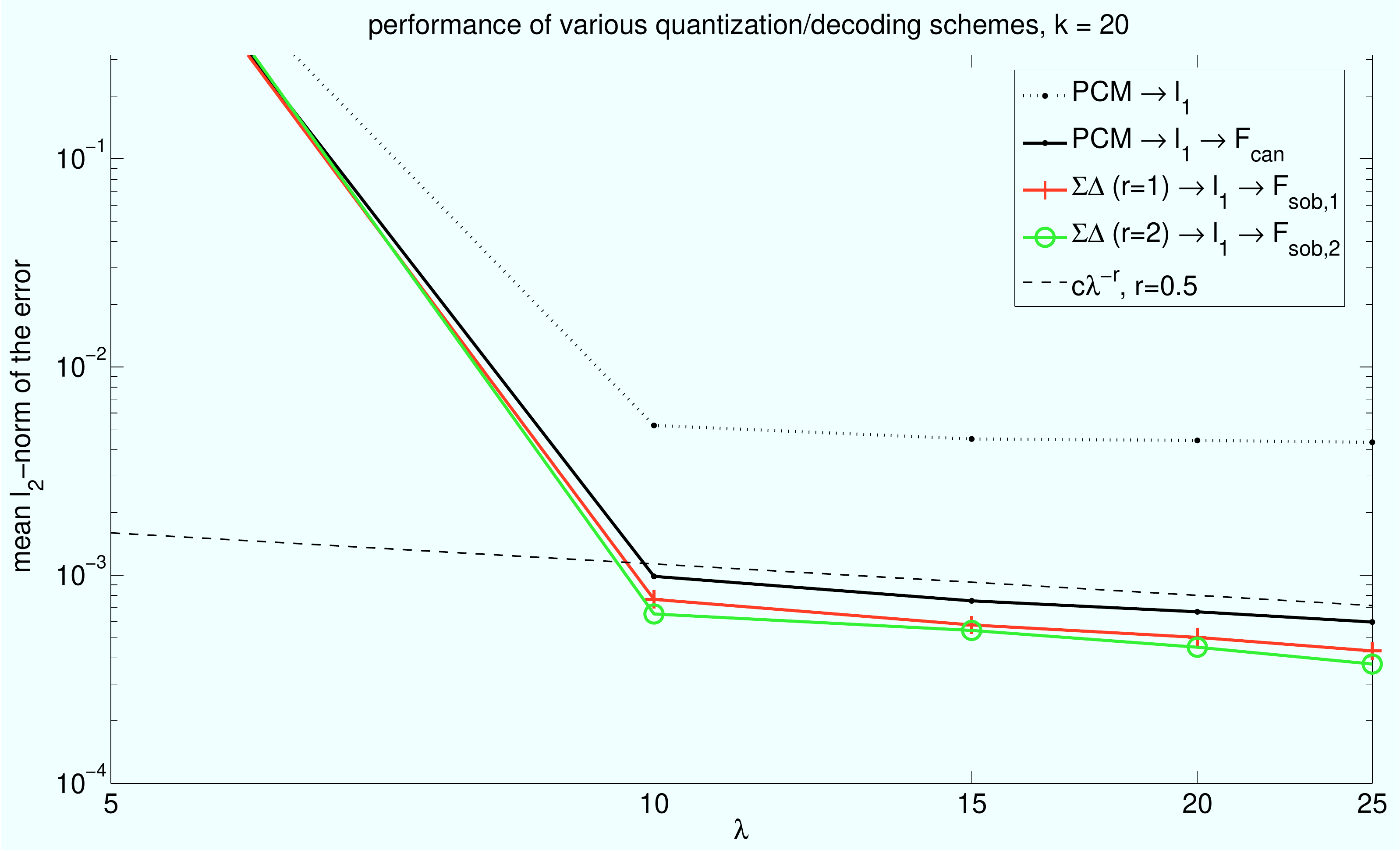}}
\caption{\label{fig6} The average case performance of the proposed $\Sigma\Delta$ quantization and reconstruction schemes (with general duals) for various values of $k$. 
For this experiment the non-zero entries of $x$ are constant and $\delta=0.01$.}
\end{center}
\end{figure}

\begin{figure}[b]
\begin{center}
\subfigure[]{
\includegraphics[width=3in,height=2in]{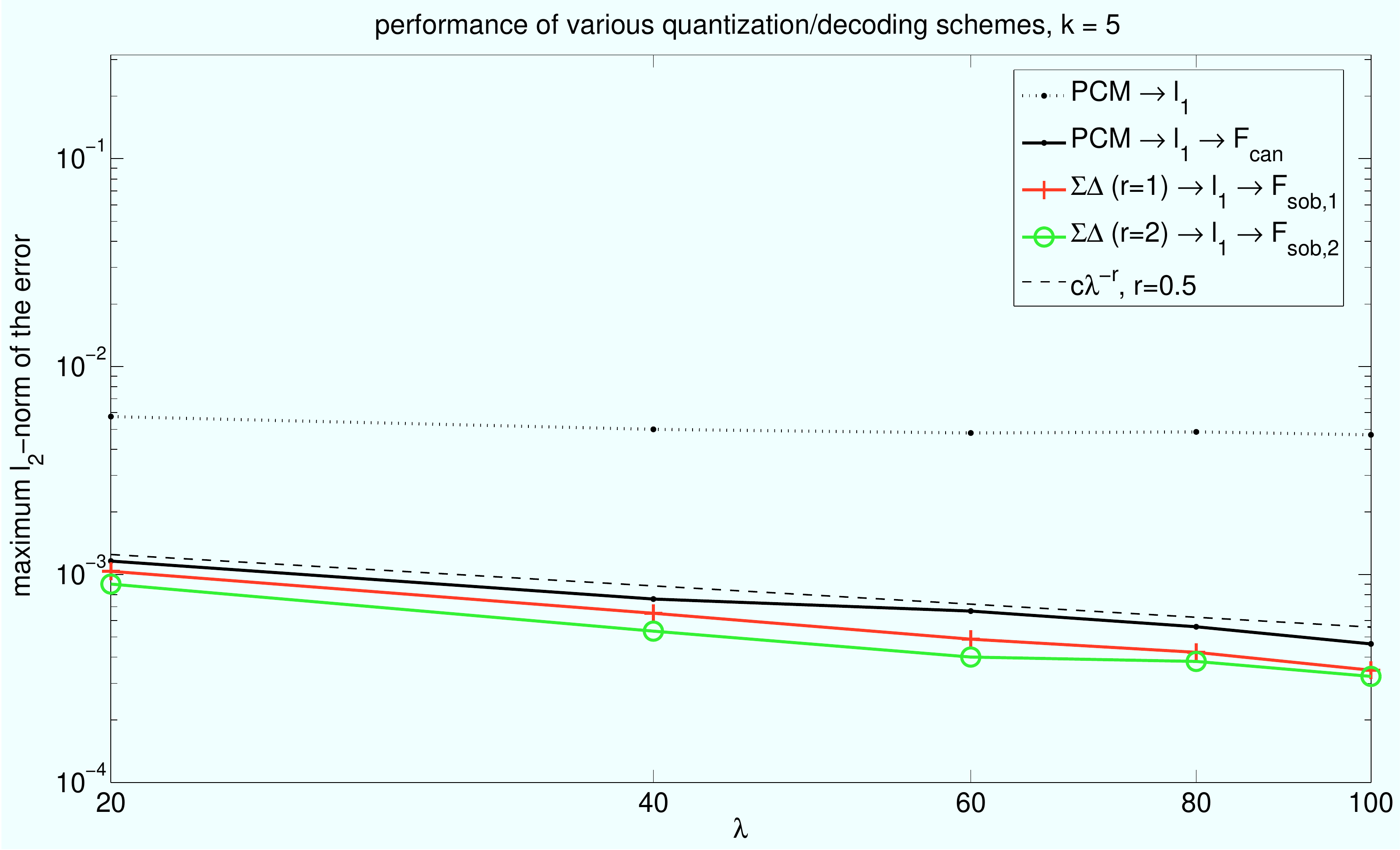}}
\subfigure[]{
\includegraphics[width=3in,height=2in]{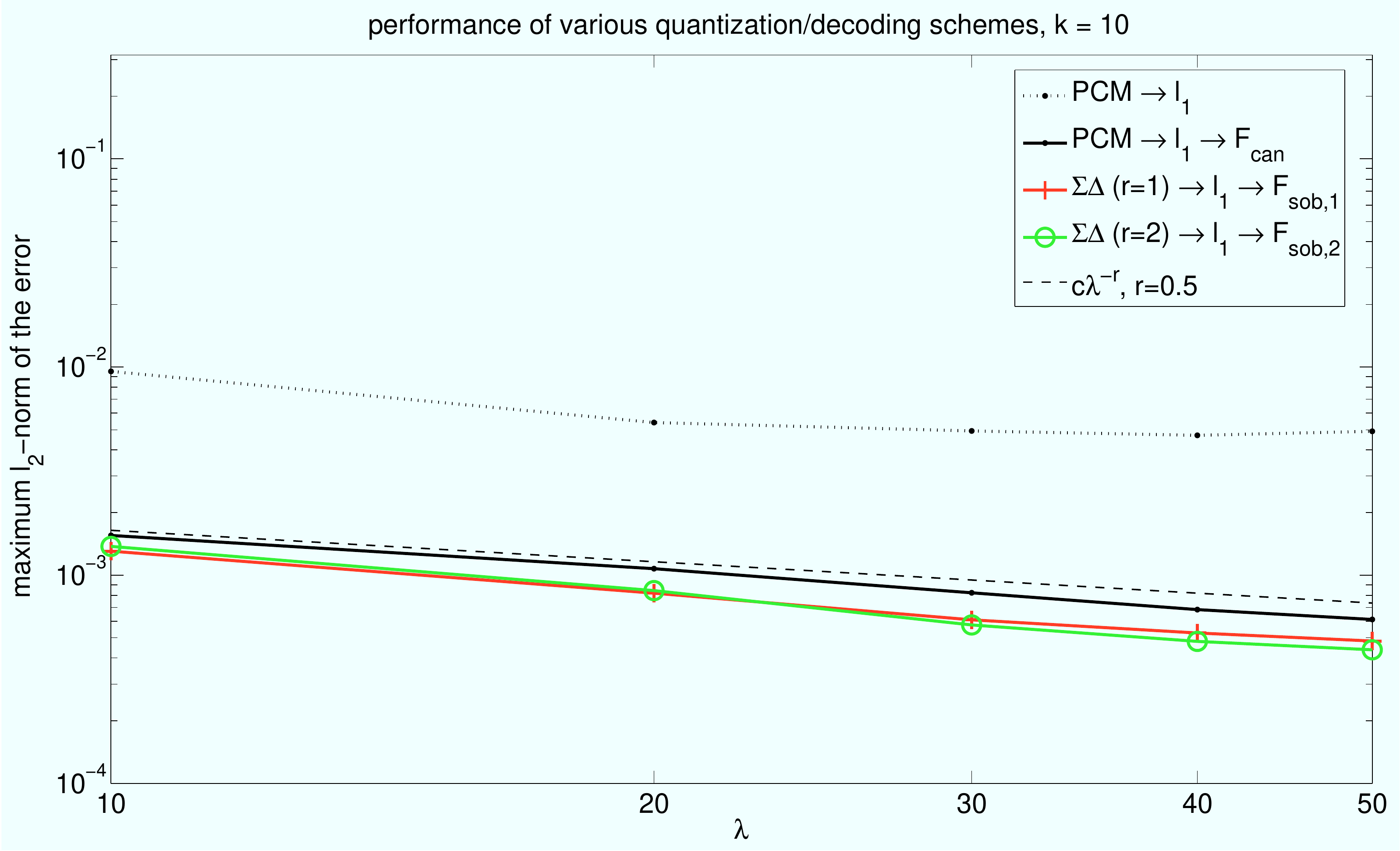}}
\subfigure[]{
\includegraphics[width=3in,height=2in]{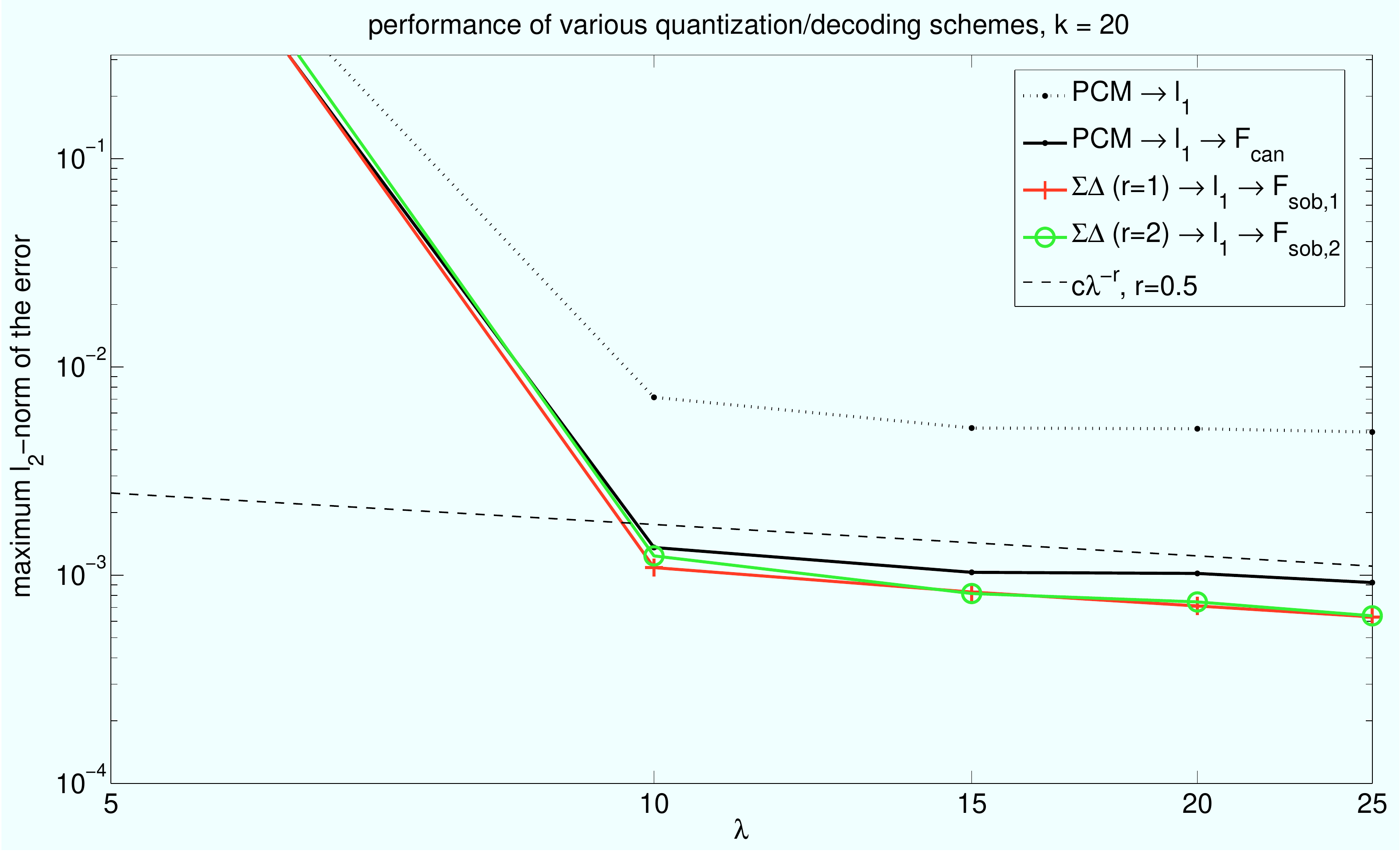}}
\caption{\label{fig7} The worst case performance of the proposed $\Sigma\Delta$ quantization and reconstruction schemes (with general duals) for various values of $k$. 
For this experiment the non-zero entries of $x$ are constant and $\delta=0.01$.}
\end{center}
\end{figure}

\section*{Acknowledgments} The authors would like to thank Ronald
DeVore for valuable discussions. This work was initiated during an AIM
workshop and matured during a BIRS workshop. We thank the American
Institute of Mathematics and Banff International Research Station for
their hospitality. This work was supported in part by: National
Science Foundation Grant CCF-0515187 (G\"unt\"urk), Alfred P. Sloan
Research Fellowship (G\"unt\"urk), National Science Foundation Grant
DMS-0811086 (Powell), a Pacific Century Graduate Scholarship from the
Province of British Columbia through the Ministry of Advanced
Education (Saab), a UGF award from the UBC (Saab), and a Natural
Sciences and Engineering Research Council of Canada Discovery Grant
(Y{\i}lmaz).

\appendix
\section{Singular values of $D^{-r}$}
\label{sec_sing_val_D}

It will be more convenient to work with the singular values of $D^r$.
Note that because of our convention of descending ordering of singular values, 
we have 
\begin{equation}\label{sing_val_inversion}
\sigma_j(D^{-r}) = \frac{1}{\sigma_{m+1-j}(D^r)},~~j=1,\dots,m.
\end{equation}
For $r=1$, an explicit formula is available \cite{von1941distribution, strang1999dct}. Indeed, we have
\begin{equation}\label{sing_val_D}
\sigma_j(D) = 2\cos\left(\frac{\pi j}{2m+1}\right),
~~j=1,\dots,m,
\end{equation}
which implies
\begin{equation}\label{sing_val_D_1}
\sigma_j(D^{-1}) = \frac{1}{2\sin\left(\frac{\pi (j-1/2)}{2(m+1/2)}\right)},
~~j=1,\dots,m.
\end{equation}

The first observation is that 
$\sigma_j(D^r)$ and $(\sigma_j(D))^r$ are different, because
$D$ and $D^*$ do not commute. However, this becomes insignificant
as $m \to \infty$. In fact, the 
asymptotic distribution of $(\sigma_j(D^r))_{j=1}^m$ as $m \to \infty$
is rather easy to find using standard results in the theory of Toeplitz matrices:
$D$ is a banded 
Toeplitz matrix whose symbol is $f(\theta) = 1 - e^{i\theta}$, hence
the symbol of $D^r$ is $(1-e^{i\theta})^r$. It then follows by 
Parter's extension of Szeg\"o's theorem \cite{Parter} that 
for any continuous function $\psi$, we have 
\begin{eqnarray} 
\lim_{m \to \infty} \frac{1}{m} \sum_{j=1}^m 
\psi( \sigma_j(D^r) ) & = & 
\frac{1}{2\pi} \int_{-\pi}^\pi \psi(|f(\theta)|^r) \,d\theta.
\end{eqnarray}
We have $|f(\theta)| = 2 \sin |\theta|/2$ for $|\theta|\leq \pi$,  
hence the distribution of $(\sigma_j(D^r))_{j=1}^m$ is
asymptotically the same as
that of $2^r \sin^r (\pi j/2m)$, and consequently, 
we can think of $\sigma_j(D^{-r})$ roughly as 
$\big(2^r \sin^r (\pi j/2m)\big)^{-1}$. Moreover, we know that
$\|D^r\|_{\mathrm{op}} \leq \|D\|^r_{\mathrm{op}} \leq 2^r$, hence
$\sigma_{\mathrm{min}}(D^{-r}) \geq 2^{-r}$.

When combined with 
known results on the rate of convergence to the limiting distribution
in Szeg\"o's theorem, the above asymptotics
could be turned into an estimate of the kind given in Proposition
\ref{sing_val_Dr}, perhaps with some loss of precision. 
Here we shall provide a more direct approach which is not asymptotic, and
works for all $m \geq 4r$. The underlying observation is that 
$D$ and $D^*$ almost commute: $D^*D  - D D^*$ has only two nonzero
entries, at $(1,1)$ and $(m,m)$. Based on this observation, 
we show below
that ${D^*}^rD^r$ is then a perturbation of 
$(D^*D)^r$ of rank at most $2r$. 

\begin{prop}\label{rank2r}
Let $C^{(r)} = {D^*}^rD^r-(D^*D)^r$
where we assume $m \geq 2r$. Define 
$$I_r :=\{1,\dots,r\}\times\{1,\dots,r\} \cup \{m-r+1,\dots,m\}\times\{m-r+1,\dots,m\}.$$ 
Then $C^{(r)}_{i,j}= 0$ for all $(i,j)\in I_r^c$. 
Therefore, $\mathrm{rank}(C^{(r)}) \leq 2r$.
\end{prop}

\begin{proof} 
Define the set $\sC_r$ of all ``$r$-cornered'' matrices as
$$\sC_r = \{M : M_{i,j} = 0 \mbox{ if } (i,j) \in I_r^c \},$$
and the set
$\sB_r$ of all ``$r$-banded'' matrices as
$$\sB_r = \{M : M_{i,j}=0 \mbox{ if } |i-j| > r \}.$$
Both sets are closed under matrix addition. 
It is also easy to check the following facts (for the admissible range 
of values for $r$ and $s$):
\begin{itemize}
\item[(i)] If $B \in \sB_r$ and $C \in \sC_s$, then 
$BC\in \sC_{r+s}$ and $CB \in \sC_{r+s}$. 
\item[(ii)] If $B \in \sB_r$ and $\tilde B \in \sB_s$, then 
$B\tilde B\in \sB_{r+s}$. 
\item[(iii)] If $C \in \sC_r$ and
$\tilde C \in \sC_s$, then $C\tilde C \in \sC_{\mathrm{max}(r,s)}$. 
\item[(iv)] If $C \in \sC_r$, then $D^* C D \in \sC_{r+1}$.
\end{itemize}

Note that $DD^*,D^*D \in \sB_1$ and 
the commutator $[D^*,D] =:\Gamma_1 \in \sC_1$.  Define
$$\Gamma_r := 
(D^*D)^r - (DD^*)^r  = (DD^*+\Gamma_1)^r - (DD^*)^r.
$$
We expand out the first term 
(noting the non-commutativity), cancel $(DD^*)^r$ and  
see that every term that remains
is a product of $r$ terms (counting each $DD^*$ as one term)
each of which is either in $\sB_1$ or in $\sC_1$.
Repeated applications of (i), (ii), and (iii)
yield $\Gamma_r \in \sC_{r}$.

We will now show by induction on $r$
that $C^{(r)} \in \sC_r$ for all $r$ such that $2r \leq m$. 
The cases $r=0$ and $r=1$ hold
trivially. Assume the statement holds for a given value of $r$. Since
$$
C^{(r+1)} = D^* (C^{(r)} + \Gamma_r) D
$$
and $\Gamma_r \in \sC_r$, property (iv) above 
now shows that $C^{(r+1)} \in \sC_{r+1}$. 
\end{proof}

The next result, originally due to 
Weyl (see, e.g., \cite[Thm 4.3.6]{HJ}), will now allow us to
estimate the eigenvalues of ${D^*}^rD^r$ using the
eigenvalues of $(D^*D)^r$:

\begin{thm}[Weyl] \label{Weyl}
Let $B$ and $C$ be $m \times m$ Hermitian matrices where $C$ has
rank at most $p$. 
Then
\begin{equation}
\lambda_{j+p}(B) \leq \lambda_{j}(B+C) \leq \lambda_{j-p}(B),~~~~
j=p+1,\dots,m-p,
\end{equation}
where we assume eigenvalues are in descending order.
\end{thm}

We are now fully equipped to prove Proposition \ref{sing_val_Dr}.

\begin{proof}[Proof of Proposition \ref{sing_val_Dr}]
We set $p=2r$, $B = (D^*D)^r$, 
and $C = C^{(r)} = {D^*}^rD^r-(D^*D)^r$ in Weyl's theorem. By Proposition
\ref{rank2r}, $C$ has rank at most $2r$. Hence, we have the relation
\begin{equation}
\lambda_{j+2r}((D^*D)^r) \leq \lambda_{j}({D^*}^rD^r) \leq 
\lambda_{j-2r}((D^*D)^r),~~~~
j=2r+1,\dots,m-2r.
\end{equation}
Since $\lambda_j((D^*D)^r) = \lambda_j(D^*D)^r$, this corresponds to
\begin{equation}\label{sing_est}
\sigma_{j+2r}(D)^r \leq \sigma_{j}(D^r) \leq 
\sigma_{j-2r}(D)^r,~~~~
j=2r+1,\dots,m-2r.
\end{equation}
For the remaining values of $j$, we will simply use the largest and smallest
singular values of $D^r$ as upper and lower bounds. However, note that
$$\sigma_1(D^r) = \|D^r\|_{\mathrm{op}} \leq 
\|D\|_{\mathrm{op}}^{r} = (\sigma_1(D))^r$$
and similarly
$$\sigma_m(D^r) = \|D^{-r}\|^{-1}_{\mathrm{op}} \geq 
\|D^{-1}\|_{\mathrm{op}}^{-r} = (\sigma_m(D))^r.$$
Hence \eqref{sing_est} can be rewritten as
\begin{equation}
\sigma_{\min(j+2r,m)}(D)^r \leq \sigma_{j}(D^r) \leq 
\sigma_{\max(j-2r,1)}(D)^r,~~~~
j=1,\dots,m.
\end{equation}
Inverting these relations via \eqref{sing_val_inversion}, we obtain
\begin{equation}
\sigma_{\min(j+2r,m)}(D^{-1})^r \leq \sigma_{j}(D^{-r}) \leq 
\sigma_{\max(j-2r,1)}(D^{-1})^r,~~~~
j=1,\dots,m.
\end{equation}
Finally, 
to demonstrate the desired bounds of Proposition \ref{sing_val_Dr}, 
we rewrite 
\eqref{sing_val_D_1} via the inequality
$2x/\pi \leq \sin x \leq x$ for $0 \leq x \leq \pi/2$ as
\begin{equation}
\frac{m+1/2}{\pi(j-1/2)} \leq \sigma_j(D^{-1}) \leq \frac{m+1/2}{2(j-1/2)},
\end{equation}
and observe that $\min(j+2r,m) \asymp_r j$ and 
$\max(j-2r,1) \asymp_r j$ for $j=1,\dots,m$. 
\end{proof}
\par {\em Remark.} The constants $c_1(r)$ and $c_2(r)$ that one 
obtains from the above argument would be significantly exaggerated.
This is primarily due to the fact that Proposition \ref{sing_val_Dr} is
not stated in the tightest possible form. The advantage of this form is 
the simplicity
of the subsequent analysis in Section \ref{lower_bound}.
Our estimates of $\sigma_{\min}(D^{-r}E)$ would become
significantly more accurate if the asymptotic distribution of 
$\sigma_j(D^{-r})$ is incorporated into our proofs in 
Section \ref{lower_bound}. However, the main disadvantage would be
that the estimates would then hold only for all sufficiently large $m$.

\frenchspacing

\bibliographystyle{plain}
\bibliography{CS_SD}

\end{document}